\pdfoutput=1
\documentclass[runningheads,orivec]{llncs}
\usepackage[T1]{fontenc}
\usepackage{graphicx}
\usepackage{hyperref}
\usepackage{comment}
\usepackage{amsmath}    %
\usepackage{xspace}     %
\usepackage{pifont}     %
\usepackage{tikz}
\usepackage{booktabs}
\usepackage{pgfplots}
\usepackage{arydshln}
\usepackage{cleveref}
\pgfplotsset{compat=1.18}
\usepackage{url}
\usepackage{amsfonts}
\usepackage{amssymb}
\usepackage[ruled,vlined,linesnumbered]{algorithm2e}
\usepackage{placeins}
\usepackage{needspace}
\usepackage{etoolbox}
\SetAlgorithmName{Algorithm}{algorithm}{List of Algorithms}
\SetAlgoCaptionSeparator{.}
\SetAlgoSkip{}
\SetAlgoInsideSkip{}    %
\SetKwProg{Alg}{}{}{}   %

\usetikzlibrary{arrows.meta,shapes,positioning}
\usepackage{color}
\definecolor{OIorange}{HTML}{E69F00}
\definecolor{OIskyblue}{HTML}{56B4E9}
\definecolor{OIgreen}{HTML}{009E73}
\definecolor{OIblue}{HTML}{0072B2}
\definecolor{OIvermillion}{HTML}{D55E00}
\definecolor{OIpurple}{HTML}{CC79A7}

\newcommand{\cmark}{\ding{51}}
\newcommand{\xmark}{\ding{55}}

\newcommand{\BG}{\ensuremath{\mathbb{G}}\xspace}

\newcommand{\BN}{\ensuremath{\mathbb{N}}\xspace}

\newcommand{\BZ}{\ensuremath{\mathbb{Z}}\xspace}

\newcommand{\CA}{\ensuremath{\mathcal{A}}\xspace}
\newcommand{\CB}{\ensuremath{\mathcal{B}}\xspace}

\newcommand{\CD}{\ensuremath{\mathcal{D}}\xspace}

\newcommand{\CQ}{\ensuremath{\mathcal{Q}}\xspace}

\newcommand{\setup}{\ensuremath{\mathsf{Setup}}\xspace}
\newcommand{\genwatermark}{\ensuremath{\mathsf{Tag}}\xspace}
\newcommand{\mergewatermarks}{\ensuremath{\mathsf{Merge}}\xspace}
\newcommand{\degrade}{\ensuremath{\mathsf{Degrade}}\xspace}
\newcommand{\trace}{\ensuremath{\mathsf{Trace}}\xspace}

\newcommand{\tid}{\ensuremath{\mathsf{id}}\xspace}

\newcommand{\sktrace}{\ensuremath{\mathsf{sk_{trace}}}\xspace}

\newcommand{\negl}{\ensuremath{\mathsf{negl}}\xspace}

\newcommand{\pp}{\ensuremath{pp}\xspace}
\newcommand{\secparam}{\ensuremath{\lambda}\xspace}

\newcommand{\etal}{\emph{et al.}\xspace}

\begin{document}
\pagestyle{plain}
\title{Track me if you can: Ephemeral coin tracing}
\author{Ignacio Amores-Sesar\inst{1}\orcidID{0000-0002-1751-1515} \and
Christian Cachin\inst{2}\orcidID{0000-0001-8967-9213} \and
Rohit Chatterjee\inst{3}\orcidID{0009-0004-1970-1779} \and
Luiza Soezima\inst{1}\orcidID{0000-0002-4766-1506} \and
François-Xavier Wicht\inst{2}\orcidID{0009-0005-6090-7901} \and
Michelle Yeo\inst{1,4}\orcidID{0009-0001-3676-4809}
}
\authorrunning{I. Amores-Sesar et al.}
\institute{Aarhus University, Denmark\\
\email{\{amores-sesar,lbrsoezima\}@cs.au.dk}
\and
University of Bern, Switzerland\\
\email{\{christian.cachin,francois-xavier.wicht\}@unibe.ch}\\
\and
National University of Singapore, Singapore\\
\email{rochat@nus.edu.sg}
\and
Nanyang Technological University, Singapore\\
michellexyeo@gmail.com}
\maketitle              %
\begin{abstract}
Privacy-preserving payment systems are well understood, yet concerns about their misuse for financial crime have led to only limited adoption in regulated settings such as central bank digital currencies (CBDCs) and institutional stablecoins. Tracing is one tool for addressing these concerns: acting on external evidence implicating a user, law enforcement follows the suspect's funds through the ledger to uncover laundering routes and accomplices. Existing coin-tracing schemes, however, provide no cryptographic bound on tracing reach: once initiated, a trace may propagate indefinitely through the transaction graph or persist across all future transactions of a targeted user. Keeping surveillance targeted and temporary therefore depends on the restraint of the authority or a committee.

We introduce \emph{ephemeral coin tracing} (ECT), a primitive that bounds tracing reach by construction. Each account carries an encrypted tag that records which traced identifiers its funds carry while hiding its tracing status from users. When funds move, the sender's tag degrades and merges with the recipient's tag. Each tracing contribution expires independently after a policy-defined number of hops and then becomes unrecoverable even to the tracing authority, without affecting other live contributions in the same tag. Public parameters also bound how many identifiers a tag can distinguish simultaneously. We formalize ECT and give constructions based on exponential ElGamal, Damg{\aa}rd--Jurik encryption, and Ring-LWE, the last providing post-quantum security.
\end{abstract}
\setcounter{tocdepth}{2}
\makeatletter
\def\l@title#1#2{}
\def\l@author#1#2{}
\makeatother
\clearpage
\tableofcontents
\clearpage
\section{Introduction}
Cryptography has long made private payments possible, from Chaum's ecash~\cite{DBLP:conf/crypto/Chaum82} to decentralized systems such as Zcash~\cite{DBLP:conf/sp/Ben-SassonCG0MTV14} and Monero~\cite{koe2020zerotomonero}.
In regulated settings, including central bank digital currencies (CBDCs) and institutional stablecoins, payment privacy must coexist with the ability to investigate financial crime.
Most regulated systems support such investigations through trusted intermediaries with broad access to transaction histories.
This access exposes the activities of ordinary users alongside those of suspects and can be abused for commercial exploitation or surveillance of political dissidents and journalists.

Investigative tools should therefore limit their reach and the information they expose. We focus on \emph{targeted tracing}, which complements suspect identification and measures such as account freezes, transaction limits, or mandatory disclosure.
Acting on outside evidence, such as a financial intelligence report, law enforcement actively embeds a \emph{tracing tag} in a suspect's funds. This tag encrypts tracing information and travels with the funds, letting the authority uncover laundering routes and accomplices.
Crucially, such tracing must not deanonymize users unrelated to the investigation.

The coin tracing schemes proposed in the literature~\cite{DBLP:conf/fc/Garman0M16,DBLP:conf/ccs/KiayiasKS22} support such targeted investigations, but grant the authority tracing capabilities that are \emph{unbounded}: once initiated, surveillance propagates indefinitely through the transaction graph, affecting potentially all users~\cite{DBLP:conf/fc/Garman0M16}, and continues through all future transactions of a targeted user~\cite{DBLP:conf/ccs/KiayiasKS22}.
These schemes provide no cryptographic guarantee that tracing remains targeted and temporary; it rests instead on the authority's goodwill.
A common mitigation distributes the capability: in PEReDi~\cite{DBLP:conf/ccs/KiayiasKS22} and UTT~\cite{DBLP:journals/iacr/TomescuBAAGPY22}, tracing requires a threshold of committee members to cooperate.
Distributing trust in this way bounds \emph{who} may trace, not \emph{how much} may be traced, and in practice the shares concentrate in a handful of entities, often state-operated, whose collusion restores unrestricted surveillance.

Secure MPC could enforce these bounds~\cite{DBLP:conf/stoc/ChaumCD88}, but running it for every transaction is too costly at ledger scale. We therefore seek lightweight tracing whose reach is fixed by public parameters and cannot be silently extended by the authority. We ask:
\begin{quote}
\emph{Can we design a simple and efficient protocol for targeted tracing whose reach is cryptographically bounded, even against the authority itself?}
\end{quote}
\subsection{Our contribution}
We introduce \emph{ephemeral coin tracing} (ECT), a primitive for targeted tracing with cryptographically enforced expiry. Its central idea is to attach encrypted tracing metadata to every account state and update that metadata as funds move. A trace follows tagged funds prospectively for a bounded number of policy-defined hops; expired contributions disappear while unrelated account data remains private.

\paragraph{Account tags and expiry.}
Every account includes an encrypted tracing tag in its authenticated state. The ledger checks that payments update these tags according to the protocol: each payment degrades the sender's tag and merges the result into the recipient's. The authority initiates a trace by issuing a tag that encrypts an identifier it assigns to the suspect; untargeted users are assigned zero. Encryption hides both the tracing status and the remaining hop budget from users. Each tracing contribution, that is, each identifier issued into a tag, expires after its assigned number of \emph{hops}, bounded by the public maximum $h$ fixed at setup. At expiry, it becomes a zero contribution that not even the authority can recover, without affecting other live contributions in the same tag.

\paragraph{Hop policy and evasion.}
The ledger policy determines when a hop occurs. It may degrade a tag on every outgoing transfer, once per epoch or block-height interval, or only when a transfer meets a condition such as a minimum amount. A user can exhaust a contribution's budget only by triggering enough such hops before moving the funds of interest. If every transfer counts, cheap ``breadcrumb'' payments may suffice; minimum amounts, rate limits, or counterparty requirements can instead make exhaustion consume time, meaningful value, or identified counterparties. Until the budget is exhausted, the trace remains live, and the evasion pattern may itself justify stronger measures or a newly recorded tag. Moving the funds immediately avoids this delay but risks making the incriminating transfer while the tag is still active.

\paragraph{Capacity and merging.}
ECT also exposes an identifier-slot bound: supporting more distinguishable identifiers requires larger public parameters and tags. This bound remains compatible with recurring investigations: a later warrant may start a new contribution, while each contribution retains its own hop budget. ECT merges tags issued at different times, and therefore at different degradation depths, homomorphically while preserving each identifier at its current depth. To the best of our knowledge, ECT is the first coin-tracing scheme to handle such merging efficiently.

\paragraph{Scope and applications.}
We use anonymous account-based payment systems as our running model~\cite{DBLP:conf/ccs/KiayiasKS22,DBLP:conf/ccs/WustKDC22}, but ECT also applies to authenticated payment states that carry, update, and merge encrypted tags, including CBDCs, regulated stablecoins, custodial wallets, and interbank networks. Since degradation is public and requires no secret, ECT is flexible: deployments can assign it to account holders, custodians, validators, or other designated parties without interaction with the authority.
\subsection{Technical approach}\label{sec:technical}
In short, ECT uses additively homomorphic encryption over a plaintext ring containing a publicly applicable nilpotent element $\nu$ of index $h$, such that $\nu^h=0$ and $\nu^{h-1}\ne0$. A tag encrypts an identifier, degradation multiplies its plaintext by $\nu$, and merging adds plaintexts. After $h$ degradations, every contribution becomes zero and remains zero, so expired contributions are neutral under merging and do not affect live ones.
\paragraph{Setup and group structure.}
Our first construction uses exponential ElGamal~\cite{DBLP:journals/tit/Elgamal85} in a cyclic group $\BG=\langle g\rangle$ of order $q^h$, where $q$ is prime and the public parameter $h$ is the hop budget. The prime-power order is essential: multiplication by $q$ in the exponent is non-invertible and erases all information after $h$ applications. We also instantiate the template with Damg{\aa}rd--Jurik encryption~\cite{DBLP:conf/pkc/DamgardJ01}, where exponentiation by the composite modulus implements degradation and base-$N$ digits separate hop depths, and with Ring-LWE, where BGV~\cite{DBLP:conf/innovations/BrakerskiGV12} ciphertexts over the ring $R=\mathbb{Z}[X,Y]/(X^{N_0}+1,Y^h)$ use powers of $Y$ for hop depth and the $X$-slots for identifiers. Multiplication by $Y$ degrades a tag, ciphertext addition merges tags, and the construction provides post-quantum security. Among these constructions, exponential ElGamal is typically the best choice for deployments that seek to cap the authority's tracing power at a small set of identifiers. We evaluate it on a commodity laptop at 128-bit security and $h=10$ hops: it adds a few kilobytes per transaction and traces 28 identifiers in 1.33~s. Deployments requiring a larger set of distinguishable identifiers can use the other constructions at the cost of larger tags: our Damg{\aa}rd--Jurik and Ring-LWE parameters support 341 and thousands of identifier slots, respectively. Any increase in capacity requires a public parameter change that the authority cannot make silently.
\paragraph{User identifiers and the base-$r$ encoding.}
User $\tid_i$ receives the plaintext identifier $m_i=r^{i-1}$, for a base $r\geq2$ satisfying $r^n<q$, so each identifier occupies one base-$r$ position and the complete encoding fits within one base-$q$ digit. To trace $\tid_i$, the authority encrypts $m_i$; a dummy tag encrypts zero.
\paragraph{Degradation via exponentiation.}
Each degradation raises the ciphertext components to the $q$-th power, changing an encryption of $m q^t$ into one of $m q^{t+1}$. The plaintext map $m\mapsto qm\bmod q^h$ is nilpotent: after $h$ applications, $q^hm\equiv0\pmod{q^h}$, and zero remains zero, so expiry is exact and permanent. Semantic security hides whether a tag contains zero or an identifier and also hides its current depth; the user therefore learns neither its status nor the remaining budget.
\paragraph{Merging across identifiers and depths.}
Multiplying ElGamal ciphertexts adds their plaintexts. Within a depth, the base-$r$ representation separates identifiers: the coefficient at position $i-1$ records the multiplicity of $\tid_i$, and $r$ is chosen to prevent overflows. Across depths, the base-$q$ representation provides a second layer of separation: a contribution $q^t m_i$ occupies digit $t$, so contributions at different depths coexist without collision. The authority recovers the traced set by reading the base-$q$ digits and then decoding each in base $r$.
\paragraph{Separation from ledger verification.}
ECT defines tag operations and their security properties; the deployment determines when and how tags enter accounts, while the ledger binds them to states, authenticates fresh entries, and prevents reuse. Limits on tagged accounts are a separate deployment policy. The scheme bounds the number of distinguishable identifier slots, but the authority assigns them: users sharing a slot are indistinguishable to $\trace$, and reassigning one before its old contributions expire risks obscuring the trail for up to $h$ hops.
\paragraph{Organization.}
Section~\ref{sec:related-work} surveys related work, and Section~\ref{sec:preliminaries} introduces the model and notation. Section~\ref{sec:tracing-scheme} defines ECT and its security properties. Section~\ref{sec:algebraic-construction} presents the nilpotent template, its three instantiations with their security proofs, and a comparison of the three. Section~\ref{sec:evaluation} evaluates our exponential ElGamal implementation. Section~\ref{sec:discussion} discusses deployment choices and applications, and Section~\ref{sec:conclusion} draws a conclusion. Appendix~\ref{app:cryptodefs} recalls standard cryptographic notions.

\section{Related work}%
\label{sec:related-work}
ECT lies at the intersection of privacy-preserving payment systems, anti-money-laundering mechanisms, and cryptographic tracing primitives. We first review approaches to reconciling transaction privacy with regulatory oversight, including controlled disclosure, privacy budgets, and collaborative analysis, and then compare ECT with prior coin-tracing, watermarking, and tracing techniques.

\paragraph{Compliant payment systems.}
Research on reconciling privacy with regulation dates back to the origins of ecash~\cite{DBLP:conf/acns/ChatzigiannisBC21}. Chaum~\cite{DBLP:conf/crypto/Chaum82} introduced blind signatures for untraceable payments, which von Solms and Naccache~\cite{DBLP:journals/compsec/SolmsN92} showed can shield extortion and other perfect crimes, and Sander and Ta-Shma~\cite{DBLP:conf/crypto/SanderT99} added a mechanism for auditing the total amount in circulation. More recent permissioned designs provide auditable token payments~\cite{DBLP:conf/aft/AndroulakiCCDET20,DBLP:conf/nsdi/NarulaVV18,DBLP:conf/ccs/CecchettiZJKJS17,DBLP:conf/esorics/ChatzigiannisB21} or add auditability to Quisquis~\cite{DBLP:conf/asiacrypt/FauziMMO19,DBLP:conf/acns/PapadoulisBGP25}, allowing selective disclosure to regulators without revealing information to other users. A parallel line of work makes the authority accountable for each decryption through on-chain logs~\cite{DBLP:conf/spw/Ryan17,DBLP:conf/isw/LiWL0G20,DBLP:journals/tifs/LiLWDWR25}, collaboration among several regulators~\cite{DBLP:journals/tifs/ZhaoluWW24}, or self-incriminating proofs that expose rogue decryptors~\cite{DBLP:journals/iacr/ChiangDFMY24}. A further line shifts compliance to the users themselves, who prove that their funds avoid illicit sources without revealing their history~\cite{DBLP:journals/bcra/ButerinINSS24}. These mechanisms constrain who may decrypt and record each decryption; ECT instead bounds the reach of the resulting trace.
\paragraph{Privacy budgets.}
Privacy-budget mechanisms enforce compliance through value thresholds, an idea that goes back to e-cash that is anonymous only up to a spending limit~\cite{DBLP:conf/scn/CamenischHL06}. A user transacts anonymously below a limit defined per transaction~\cite{DBLP:conf/fc/WustKCC19}, per period~\cite{DBLP:journals/iacr/TomescuBAAGPY22}, or for incoming value~\cite{DBLP:conf/ccs/WustKDC22}; above the limit, disclosure to the regulator is mandatory. The same threshold applies to all users. Thus, a user with high transaction volume, such as a currency exchange, may be exposed indefinitely. Moreover, inflation can cause an increasing fraction of users to exceed a fixed threshold without any change in policy. ECT instead bounds each trace in hops rather than in value.
\paragraph{Collaborative anti-money-laundering.}
Money-laundering routes often cross institutional boundaries, but banks generally cannot pool their records. Secure risk propagation over an interbank network~\cite{DBLP:conf/fc/EgmondDBRSPV24}, AMLChain~\cite{DBLP:conf/eisa/HeC21}, and homomorphic graph inference~\cite{DBLP:conf/space/EffendiC24} privately compute risk scores or compliance predicates from records held by several institutions. ECT is complementary: once such an analysis, or conventional off-ledger evidence, identifies a target, ECT prospectively follows the associated funds across participating banks. It does not require a joint computation over the entire interbank graph on the critical path of every payment.
\paragraph{Coin tracing.}
Fair e-cash lets designated trustees lift the anonymity of otherwise untraceable payments~\cite{DBLP:conf/eurocrypt/StadlerPC95,DBLP:conf/soda/BrickellGK95,DBLP:conf/esorics/CamenischMS96,DBLP:conf/asiacrypt/FrankelTY96,DBLP:conf/fc/DavidaFTY97,DBLP:conf/ccs/JakobssonY96}, which bounds who may revoke but not how much follows once revocation is triggered. Garman \etal~\cite{DBLP:conf/fc/Garman0M16} mark individual coins in decentralized anonymous payment schemes~\cite{DBLP:conf/sp/Ben-SassonCG0MTV14}: traceability propagates to every descendant coin until a regulated entity clears it, so a trace has no horizon of its own and the tracing payload grows linearly in the number of inputs. Kiayias \etal~\cite{DBLP:conf/ccs/KiayiasKS22} and Tomescu \etal~\cite{DBLP:journals/iacr/TomescuBAAGPY22} improve accountability with a committee-initiated scheme in which each of a user's transactions carries a pseudorandom tag derived from a hidden per-user tracing key, which the committee reconstructs from shares to locate those transactions; a trace, once started, still covers every transaction the user makes. ECT differs from both in bounding the reach of a trace cryptographically rather than procedurally.
\paragraph{Watermarking and tracing primitives.}
Tag embedding resembles cryptographic watermarking, and more broadly steganography~\cite{DBLP:conf/ih/Cachin98}, but a watermark is meant to persist, whereas our tags must merge with one another and expire at a fixed budget. Undetectable watermarks for language models~\cite{DBLP:conf/colt/ChristGZ24} embed a single bit for a single user; Cohen \etal~\cite{DBLP:conf/sp/Cohen0S25} reach the multi-user setting but not merging, for the reason given in Section~\ref{sec:algebraic-construction}; Zhao \etal~\cite{DBLP:conf/sp/ZhaoGCFFCG0NTJL25} survey the area. Traitor tracing~\cite{DBLP:journals/tit/ChorFNP00,DBLP:conf/eurocrypt/BonehSW06,DBLP:conf/crypto/BonehPR24a,DBLP:conf/eurocrypt/Zhandry25a} and traceable signatures~\cite{DBLP:conf/eurocrypt/KiayiasTY04,DBLP:conf/ctrsa/Fujisaki11,DBLP:conf/latincrypt/Ghadafi14} trace the party responsible for a forbidden or disputed action, whereas ECT traces the ordinary transactions of a user designated in advance.
\section{Preliminaries}
\label{sec:preliminaries}
\paragraph{Notation.}
For $n \in \mathbb{N}$, let $[n]=\{1,\dots,n\}$. We write $x \xleftarrow{\$} S$ when $x$ is sampled uniformly from $S$. A function $f:\mathbb{N}\to\mathbb{R}$ is negligible, written $f(\secparam)=\negl(\secparam)$, if $f(\secparam)<1/p(\secparam)$ for every polynomial $p$ and all sufficiently large $\secparam$. PPT denotes probabilistic polynomial time.

To support tag merging in the plaintext ring, we use the following positional encoding: a multiset is represented by an integer whose base-$r$ digits record element multiplicities, allowing the multiset to be recovered after encoded elements are added.
\begin{definition}[Base-$r$ encoding]
  \label{def:base-r-encoding}
  Let $n \in \BN$ and $r \geq 2$. Assign index $i \in [n]$ the value $m_i=r^{i-1}$. For a multiset $T$ over $[n]$ with multiplicities $c(i)\in\{0,\ldots,r-1\}$, define $v_T=\sum_{i=1}^n c(i)r^{i-1}$. Because $c(i)<r$, no digit carries into the next position. Hence $v_T\leq r^n-1$ and uniquely determines $T$.
\end{definition}
When tracing tags are merged, the coefficient of each base-$r$ position records the multiplicity of the corresponding identifier. These coefficients must remain below $r$ to prevent overflows into adjacent positions; the ledger's bound on transaction fan-out ensures that they do.
\begin{lemma}[Bounded-fan-out multiplicity]
  \label{lem:multiplicity}
  Suppose a transaction copies each input-tag contribution to at most $d$ output tags. Then an identifier has multiplicity at most $d^t$ at depth $t$, and at most $d^{h-1}$ at any live depth $t<h$.
\end{lemma}
\begin{proof}
  The multiplicity is one at depth $0$. If it is at most $d^t$ at depth $t$, copying each occurrence to at most $d$ outputs gives multiplicity at most $d^{t+1}$ at depth $t+1$. Induction proves the first bound; the maximum over $t<h$ gives the second. \qed
\end{proof}
In our two-state account model, $d=2$: the sender's tag propagates to both the change state and the recipient's new state, giving multiplicity at most $2^{h-1}$. An incoming payment carries the recipient's existing contributions to only one successor state. Thus, only outgoing transfers increase multiplicity (Section~\ref{sec:tracing-scheme}).
\subsection{Cryptographic notions}
The ECT primitive uses additively homomorphic encryption and signatures. We require signature correctness and existential unforgeability under chosen-message attacks (EUF-CMA). A deployment may additionally use a NIZK, direct public verification, or a trusted-execution attestation to enforce its state-transition policy, but that verification mechanism is not part of the ECT primitive. Appendix~\ref{app:cryptodefs} recalls the cryptographic notions used in the paper. We state homomorphic encryption explicitly because its algebraic properties determine merging and degradation.
\paragraph{Homomorphic encryption.}
An \emph{additively homomorphic encryption scheme} over a message group $\mathbb{M}$ consists of PPT algorithms $(\mathsf{KeyGen}, \mathsf{Enc}, \mathsf{Dec}, \mathsf{Rerand})$ satisfying correctness and IND-CPA security. The ciphertext space forms a group under $\cdot$, and for $c_i=\mathsf{Enc}(\mathsf{pk},m_i)$,
$\mathsf{Dec}(\mathsf{sk},c_1\cdot c_2)=m_1+m_2$. Rerandomization computes $c\cdot\mathsf{Enc}(\mathsf{pk},0)$ and produces a ciphertext computationally indistinguishable from a fresh encryption of the same message. Homomorphic addition implements merging; rerandomization hides the link between a tag and its degraded successor.
\subsection{System model}
\label{sec:system-model}
ECT is a cryptographic layer attached to a host ledger. The model separates the primitive's tag operations from the deployment mechanisms that authenticate account states, enforce transitions, and determine when tags are issued or degraded.

\paragraph{Ledger and accounts.}
All parties are computationally bounded. The system consists of users, a tracing authority holding $\sktrace$, and a permissioned ledger with at least $3f+1$ validators, of which at most $f$ are Byzantine. Users enroll under real-world identities and receive credentials. The authority knows these identities but cannot link them to ledger states except at system boundaries. As in Platypus~\cite{DBLP:conf/ccs/WustKDC22} and PayOff~\cite{DBLP:conf/asiaccs/BeerZKWCC26}, an account state commits to a balance, account secret, tracing tag, and epoch. A payment consumes the sender's and recipient's states and creates two new states. Its proof establishes membership, correct nullifiers, balance conservation, enrollment, and the required tag transition. Each tag is a ciphertext under the authority's key and is bound to its published state.
\paragraph{Tag lifecycle.}
Time is divided into epochs. At the start of each epoch, the authority publishes a signed \emph{tag board} with one new ciphertext for every enrolled identity. The entry encrypts zero unless the authority issues or renews a warrant for that identity in this epoch, in which case it encrypts a tracing identifier assigned to that identity; semantic security hides which entries are nonzero. These authenticated board entries are the only source of new tag contributions. Each account merges its board entry into its tag, so that existing and newly issued contributions remain active together; under our default policy, the account first degrades its existing tag, placing contributions issued in different epochs at different depths. The transaction proof hides which board entry was selected while enforcing this update rule and its binding to the account owner's credential; a per-identity, per-epoch nullifier prevents the same entry from being ingested more than once. The policy also determines when accounts must ingest entries and whether dormant accounts consume hop budget. To renew a warrant, the authority issues a fresh identifier in a later epoch, that is, one never issued before or one whose earlier contributions have expired, for instance because every account degrades at each ingestion and $h$ epochs have passed. The authority records the identity behind each identifier, so renewal does not affect attribution. Each identifier therefore has at most one live issuance, as Lemma~\ref{lem:multiplicity} requires. The warrant controls when fresh contributions may be issued, while ECT limits how far each contribution can propagate.
\paragraph{Privacy and adversaries.}
We assume that hiding commitments and membership proofs over the full state set provide ledger confidentiality, unlinkability, and untraceability~\cite{DBLP:conf/ccs/KiayiasKS22,DBLP:conf/fc/WichtWLC24}. ECT reveals only live tags within their hop window; dummy and expired tags decrypt to zero. Identities are recovered only at system boundaries or through later enforcement. The primitive guarantees that no party, including the tracing authority, can recover a contribution after its hop budget and that a tag can distinguish at most $\lfloor\log_r q\rfloor$ identifier slots without larger parameters. The ledger integration is responsible for rejecting malformed transitions, binding tags to states, preventing reuse, and enforcing any issuance quota or slot-uniqueness policy. We exclude out-of-band disclosure of tracing results. The two-state account model supplies the fan-out bound required by Lemma~\ref{lem:multiplicity}.

With this model in place, we are ready to define the ephemeral tracing scheme.
\section{Ephemeral tracing scheme}
\label{sec:tracing-scheme}
An ephemeral tracing scheme tracks identifiers through encrypted tags that move toward expiry whenever they are degraded. We call each application of degradation a \emph{hop}. ECT specifies the effect of a hop and the number of hops a tag survives, but the ledger policy determines which events trigger one. A hop might be charged for every outgoing transfer, only for transfers above a value threshold, once per epoch or block-height interval, or upon interaction with an enrolled counterparty. This choice determines how costly it is to evade tracing: charging every transfer may let a target exhaust the budget through many negligible payments, whereas thresholds, rate limits, or counterparty requirements can make exhaustion consume time, meaningful value, or identified counterparties (see Section~\ref{sec:technical}). A tag's \emph{hop budget} is the number of such policy-defined steps it survives; once that budget is exhausted, the tag becomes untraceable even to the authority. When several payment states are combined, their tags can be merged without resetting their progress toward expiry, so the resulting tag preserves each live identifier at its current degradation depth.
\subsection{Characterization}
An ephemeral tracing scheme consists of the following six algorithms:
\begin{description}
  \item[$\setup(1^\secparam, n, h, r) \rightarrow (\pp, \sktrace)$]
    Takes as input the security parameter $\secparam$, the number of identifier slots $n$ (the maximum number of identifiers that a tag can distinguish at once), the \emph{maximum} hop budget $h$, and the encoding base $r$.
    Outputs public parameters $\pp$, which fix $h$ publicly, and a tracing secret key $\sktrace$.
    The public parameters \pp are implicit inputs to the remaining algorithms.
  \item[$\genwatermark(\sktrace, \tid, b, h') \rightarrow (w, \sigma)$]
    Takes as input the tracing key, target identifier $\tid$, trace flag
    $b \in \{0,1\}$, and a per-tag hop budget $h' \le h$.
    Returns a tag $w$ embedding $\tid$ if $b = 1$ or zero if $b = 0$, together with an authority signature $\sigma$ on $w$. The tag expires after exactly $h'$ degradation hops.
  \item[$\mathsf{Verify}(w,\sigma) \rightarrow \{0,1\}$]
    Verifies that $\sigma$ is a valid authority signature on the freshly issued tag $w$.
  \item[$\degrade(w) \rightarrow (w',\xi)$]
    Takes as input a tag $w$.
    Returns a rerandomized tag $w'$ after one degradation step and the private transition witness $\xi$.
  \item[$\mergewatermarks(w_1, \ldots, w_m) \rightarrow w'$]
    Takes as input tags $w_1, \ldots, w_m$.
    Returns the merged tag $w' = \prod_j w_j$, which contains the union of the traced identifiers across all inputs.
  \item[$\trace(\sktrace, w) \rightarrow S$ or $\bot$]
    Takes as input the tracing key and a tag $w$.
    Outputs a set $S$ of recovered identifiers, or $\bot$ if the coin is untraced (dummy tag) or fully degraded.
\end{description}
The authority runs \setup, \genwatermark, and \trace, while validators use $\mathsf{Verify}$ to authenticate freshly issued tags. In contrast, \degrade and \mergewatermarks are public operations that require no secret key, so the ledger policy may assign them to account holders, custodial intermediaries, validators, or another designated party. Under our default payment policy, an outgoing payment advances the contributions in the sender's tag by one hop but does not advance contributions already held by the recipient. Given the sender's tag $w_S$ and the recipient's tag $w_R$, the party constructing the transaction computes $(\hat w_S,\xi)\leftarrow\degrade(w_S)$. It assigns $\hat w_S$ to the sender's successor state and merges $\hat w_S$ with $w_R$ for the recipient's successor state; that is, $w'_S=\hat w_S$ and $w'_R=\mergewatermarks(\hat w_S,w_R)$. The party retains $\xi$ as a witness for the host transaction proof, which binds the tag update to the hidden input state. Alternatively, a trusted execution environment may attest the update, or validators may apply a deterministic, publicly recomputable degradation. Charging a hop only to the sender's tag prevents third parties from exhausting a target's budget through unsolicited incoming payments. Under the default epoch-ingestion policy, an account similarly degrades its existing tag before merging its new board entry, placing contributions issued in successive epochs at different depths (Section~\ref{sec:system-model}).
Figure~\ref{fig:example} illustrates propagation, merging, and per-contribution expiry under the default policy in which spending advances the sender's tag.
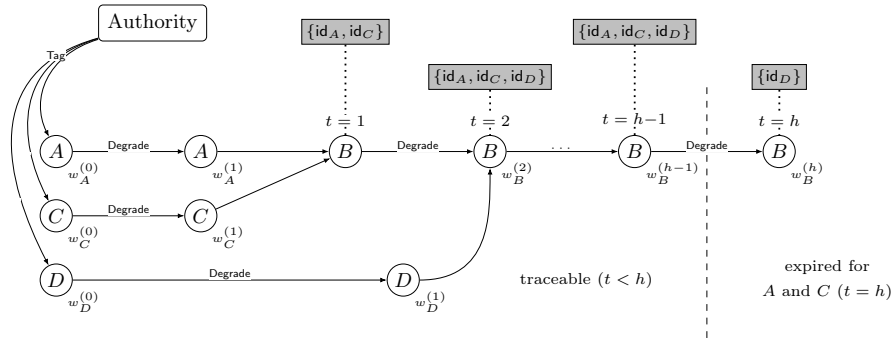
\begin{figure}[ht!]
    \centering
    \resizebox{\linewidth}{!}{%
    \begin{tikzpicture}[
      >=Stealth,
      xscale=1.5,
      state/.style={draw, circle, minimum size=0.5cm, inner sep=1pt, font=\bfseries, fill=white},
      taglbl/.style={font=\tiny, align=center},
      edgelbl/.style={font=\tiny\itshape, inner sep=0pt, fill=white},
      traceanno/.style={rotate=0,font=\scriptsize, fill=lightgray, draw=black, rectangle, inner sep=2pt, minimum size=0.0cm},
      authbox/.style={draw, rectangle, rounded corners=2pt, inner sep=4pt, font=\small},
    ]
    \node[authbox] (auth) at (1,2) {Authority};
    \node[state] (A0) at (0,0) {$A$};
    \node[state] (C0) at (0,-1) {$C$};
    \node[state] (D0) at (0,-2) {$D$};
    \node[taglbl, below right=-5pt of A0] {$w_A^{(0)}$};
    \node[taglbl, below right=-5pt of C0] {$w_C^{(0)}$};
    \node[taglbl, below right=-5pt of D0] {$w_D^{(0)}$};
    \draw[-{Latex[length=1mm]}] (auth) to[in=110, out=200] node[edgelbl, near end, left] {} (C0);
    \draw[-{Latex[length=1mm]}] (auth) to[in=110, out=200] node[edgelbl, midway, near end, left] {} (D0);
    \draw[-{Latex[length=1mm]}, font=\small] (auth) to[in=110, out=200] node[edgelbl, near start, left, fill=white] {$\genwatermark$} (A0);
    \node[state] (A1) at (1.5,0) {$A$};
    \node[state] (C1) at (1.5,-1) {$C$};
    \node[state] (D1) at (3.6,-2) {$D$};
    \node[taglbl, below right=-5pt of A1] {$w_A^{(1)}$};
    \node[taglbl, below right=-5pt of C1] {$w_C^{(1)}$};
    \node[taglbl, below right=-5pt of D1] {$w_D^{(1)}$};
    \draw[-{Latex[length=1mm]}] (A0) -- node[edgelbl, above] {$\degrade$} (A1);
    \draw[-{Latex[length=1mm]}] (C0) -- node[edgelbl, above] {$\degrade$} (C1);
    \draw[-{Latex[length=1mm]}] (D0) -- node[edgelbl, above] {$\degrade$} (D1);
    \node[state] (B1) at (3.0,0) {$B$};
    \node[traceanno, above=40pt of B1] (setB1) {$\{\mathsf{id}_A,\mathsf{id}_C\}$};
    \node[state] (B2) at (4.5,0) {$B$};
    \node[state] (Bk) at (6.0,0) {$B$};
    \node[state] (Bk1) at (7.5,0) {$B$};
    \draw[-{Latex[length=1mm]}] (A1) to (B1);
    \draw[-{Latex[length=1mm]}] (C1) to (B1);
    \draw[-{Latex[length=1mm]}] (D1) to[in=270, out=0] (B2);
    \node[traceanno, above=20pt of B2] (setB2) {$\{\mathsf{id}_A,\mathsf{id}_C,\mathsf{id}_D\}$};
    \node[traceanno, above=40pt of Bk] (setBk) {$\{\mathsf{id}_A,\mathsf{id}_C,\mathsf{id}_D\}$};
    \node[traceanno, above=20pt of Bk1] (setBk1) {$\{\mathsf{id}_D\}$};
    \draw[-, dotted, thick] (setB1) -- (B1);
    \draw[-, dotted, thick] (setB2) -- (B2);
    \draw[-, dotted, thick] (setBk) -- (Bk);
    \draw[-, dotted, thick] (setBk1) -- (Bk1);
    \draw[dashed] (6.75,1) -- (6.75,-3.0);
    \node[taglbl, below right=-5pt of B2] {$w_B^{(2)}$};
    \node[taglbl, below right=-5pt of Bk] {$w_B^{(h-1)}$};
    \node[taglbl, below right=-5pt of Bk1] {$w_B^{(h)}$};
    \draw[-{Latex[length=1mm]}] (B1) -- node[edgelbl, above] {$\degrade$} (B2);
    \draw[-{Latex[length=1mm]}] (B2) -- node[edgelbl, above] {$\cdots$} (Bk);
    \draw[-{Latex[length=1mm]}] (Bk) -- node[edgelbl, above] {$\degrade$} (Bk1);
    \node[above=4pt of B1, font=\scriptsize\itshape,fill=white,inner sep=0pt, minimum size=0cm] {$t=1$};
    \node[above=4pt of B2, font=\scriptsize\itshape,fill=white,inner sep=0pt, minimum size=0cm] {$t=2$};
    \node[above=4pt of Bk, font=\scriptsize\itshape,fill=white,inner sep=0pt, minimum size=0cm] {$t=h{-}1$};
    \node[above=4pt of Bk1,font=\scriptsize\itshape,fill=white,inner sep=0pt, minimum size=0cm] {$t=h$};
    \node at (5.5,-2.0) {\scriptsize traceable ($t < h$)};
    \node[align=center] at (8,-2.0) {\scriptsize expired for \\ \scriptsize $A$ and $C$ ($t = h$)};
    \end{tikzpicture}}
    \caption{Tag propagation, merging, and expiry, with $t$ counting hops and $h$ the budget. The authority issues tags with $\genwatermark$, and every outgoing payment applies $\degrade$ to the spent state. At $t=1$, $A$ and $C$ pay $B$, and their degraded tags merge into $B$'s state; from then on $B$ spends once per step, advancing every contribution it carries. $D$ stays dormant until it pays $B$ at $t=2$, so its contribution joins one hop younger and a single tag then carries three identifiers at two depths. The gray boxes give the set returned by $\trace$. The contributions of $A$ and $C$ reach depth $h$ at $t=h$ and decrypt to zero, while that of $D$ remains recoverable: expiry is per contribution and exact at the budget.}
    \label{fig:example}
\end{figure}
Figure~\ref{fig:protocol-overview} summarizes the resulting flow between the authority, users, and validators. This assignment is illustrative: ECT specifies the tag operations but leaves transition enforcement to the deployment.
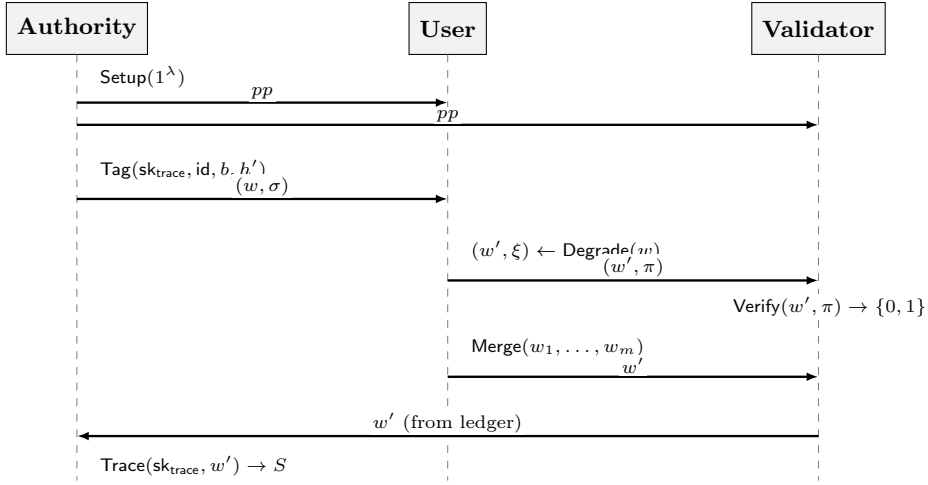
\begin{figure}[ht!]
\centering
\resizebox{\linewidth}{!}{%
\begin{tikzpicture}[
  >=Stealth,
  xscale=2.0,
  party/.style={draw, rectangle, minimum width=0pt, minimum height=0.7cm,
                inner sep=4pt, font=\small\bfseries, fill=gray!10},
  alabel/.style={font=\scriptsize\itshape, text=black},
  arrlbl/.style={font=\scriptsize, fill=white, inner sep=1pt},
  msg/.style={-{Latex[length=1.6mm]}, line width=0.9pt},
]
\def\xauth{0}
\def\xuser{2.5}
\def\xval{5}
\node[party] at (\xauth, 0) {Authority};
\node[party] at (\xuser, 0) {User};
\node[party] at (\xval,  0) {Validator};
\draw[dashed, gray] (\xauth, -0.35) -- (\xauth, -6.1);
\draw[dashed, gray] (\xuser, -0.35) -- (\xuser, -6.1);
\draw[dashed, gray] (\xval,  -0.35) -- (\xval,  -6.1);
\node[alabel, right] at (\xauth+0.1, -0.65) {$\setup(1^\secparam)$};
\draw[msg] (\xauth, -1.0) -- node[arrlbl, above]{$\pp$} (\xuser, -1.0);
\draw[msg] (\xauth, -1.3) -- node[arrlbl, above]{$\pp$} (\xval,  -1.3);
\node[alabel, right] at (\xauth+0.1, -1.9) {$\genwatermark(\sktrace, \tid, b, h')$};
\draw[msg] (\xauth, -2.3) -- node[arrlbl, above]{$(w, \sigma)$} (\xuser, -2.3);
\node[alabel, right] at (\xuser+0.1, -3.0) {$(w',\xi)\leftarrow\degrade(w)$};
\draw[msg] (\xuser, -3.4) -- node[arrlbl, above]{$(w',\pi)$} (\xval, -3.4);
\node[alabel, left, fill=white, inner sep=1pt] at (\xval+0.75, -3.75)
    {$\mathsf{Verify}(w',\pi) \to \{0,1\}$};
\node[alabel, right] at (\xuser+0.1, -4.3) {$\mergewatermarks(w_1, \ldots, w_m)$};
\draw[msg] (\xuser, -4.7) -- node[arrlbl, above]{$w'$} (\xval, -4.7);
\draw[msg] (\xval, -5.5) -- node[arrlbl, above]{$w'$ (from ledger)} (\xauth, -5.5);
\node[alabel, right] at (\xauth+0.1, -5.9) {$\trace(\sktrace, w') \to S$};
\end{tikzpicture}}
\caption{Protocol overview: interactions between the authority, users, and validators.
The authority runs \setup and \genwatermark (producing $(w, \sigma)$) and \trace (requiring $\sktrace$).
Users run \degrade and \mergewatermarks when constructing payments.
Validators run the host system's $\mathsf{Verify}$ algorithm on the transaction-validity proof, whose witness includes $\xi$.}
\label{fig:protocol-overview}
\end{figure}
\subsection{Security}
\label{subsec:security-defs}
\label{app:defs}
We require two correctness conditions and four security properties, which we define formally below. \emph{Trace correctness} (Definition~\ref{def:trace-correctness}) requires a live tag to reveal its identifier. \emph{Merge correctness} (Definition~\ref{def:merge-correctness}) requires a merged tag to reveal exactly its live contributions; dummy and expired contributions must be neutral. \emph{Indistinguishability} (Definition~\ref{def:indist}) prevents a user with access to the tagging oracle from distinguishing a traced tag from a dummy tag. \emph{Ephemeral traceability} (Definition~\ref{def:ephemeral}) makes expiry irrevocable even for an authority holding \sktrace. \emph{Unforgeability} (Definition~\ref{def:unforge}) prevents an adversary from framing an identifier that the authority never tagged; because encryption is public, the notion only concerns tags with a valid issuance and transition history. \emph{Unlinkability} (Definition~\ref{def:unlink}) prevents correlation across degradation steps.
For the games below, \setup produces public parameters $\pp$ and tracing key
$\sktrace$. The authority calls
$(w,\sigma)\leftarrow\genwatermark(\sktrace,\tid,b,h')$, where $b=1$ embeds
$\tid$, $b=0$ embeds zero, and $h'\le h$ is the contribution's budget;
$\mathsf{Verify}(w,\sigma)$ authenticates this fresh issuance. Public
degradation returns $(w',\xi)\leftarrow\degrade(w)$, with private witness
$\xi$, public merging returns $\mergewatermarks(w_1,\ldots,w_m)$, and
$\trace(\sktrace,w)$ returns the live identifiers or $\bot$. We suppress the
public identifier-slot bound $n$, maximum budget $h$, and encoding base $r$,
and write $\degrade^k(w)$ for $k$ sequential applications while omitting the
retained witnesses. Under the two-output payment policy, a contribution occurs
at most $2^t$ times at depth $t$; hence $r>2^{h-1}$ prevents carries at every
live depth and makes the identifier multiset uniquely decodable.

The following games define trace correctness, merge correctness, indistinguishability, ephemeral traceability, unforgeability, and unlinkability, in that order.
Trace correctness captures the basic functionality of the scheme: a tag issued for an identifier and degraded within its hop budget must remain recoverable by the authority.
\begin{definition}[Trace correctness]
  \label{def:trace-correctness}
  An ECT scheme satisfies \emph{trace correctness} if for every identifier $\tid$, every per-tag budget $h' \le h$, and every $k < h'$:
  \[
    \Pr\left[
      \begin{array}{l}
        (\pp, \sktrace) \leftarrow \setup(1^\secparam); \\
        (w^{(0)}, \sigma) \leftarrow \genwatermark(\sktrace, \tid, 1, h'); \\
        w^{(k)} \leftarrow \degrade^k(w^{(0)}); \\
        \tid \in \trace(\sktrace, w^{(k)})
      \end{array}
    \right] \geq 1 - \negl(\secparam),
  \]
  where $\degrade^k(w^{(0)})$ denotes $k$ sequential applications of the $\degrade$ operation.
\end{definition}
Merge correctness extends trace correctness to the multi-coin setting: after merging, the authority must recover exactly the identifiers whose contributions are still within budget. Expired and dummy contributions must act as neutral elements of merging, vanishing without disturbing live ones.
\begin{definition}[Merge correctness]
  \label{def:merge-correctness}
  An ECT scheme satisfies \emph{merge correctness} if for every $m \in \BN$, identifiers $\tid_1, \ldots, \tid_m$, flags $b_1, \ldots, b_m \in \{0,1\}$, per-tag budgets $h'_1, \ldots, h'_m \le h$, and degradation counts $t_1, \ldots, t_m \geq 0$, writing $L = \{i : b_i = 1 \wedge t_i < h'_i\}$ for the live inputs and $S_L = \{\tid_i : i \in L\}$ if $L \neq \emptyset$ and $S_L = \bot$ otherwise, the following holds whenever each identifier occurs at most $2^{h-1}$ times among the live inputs at any one depth $h-h'_i+t_i$:
  \[
    \Pr\left[
      \begin{array}{l}
        (\pp, \sktrace) \leftarrow \setup(1^\secparam); \\
        (w_i^{(0)}, \sigma_i) \leftarrow \genwatermark(\sktrace, \tid_i, b_i, h'_i), \; i \in [m]; \\
        w_i \leftarrow \degrade^{t_i}(w_i^{(0)}), \; i \in [m]; \\
        w' \leftarrow \mergewatermarks(w_1, \ldots, w_m); \\
        \trace(\sktrace, w') = S_L
      \end{array}
    \right] \geq 1 - \negl(\secparam),
  \]
  where $\degrade^{t_i}(w_i^{(0)})$ denotes $t_i$ sequential applications of $\degrade$. The multiplicity condition follows from a valid history under the two-output payment policy with at most one live issuance per identifier (Lemma~\ref{lem:multiplicity}); it is necessary because arbitrary repetitions can overflow an identifier slot.
\end{definition}
Indistinguishability ensures that a user cannot determine whether their coin carries a traceable or a dummy tag, even given arbitrarily many other tags of its choice.
\begin{definition}[Indistinguishability]
  \label{def:indist}
  An ECT scheme satisfies \emph{indistinguishability} if for every PPT adversary $\CA$:
  \[
    \Pr\left[
      \begin{array}{l}
        (\pp, \sktrace) \leftarrow \setup(1^\secparam); \\
        (\tid, h') \leftarrow \CA^{\genwatermark(\sktrace,\,\cdot,\,\cdot,\,\cdot)}(\pp); \\
        b \xleftarrow{\$} \{0,1\}; \\
        (w, \sigma) \leftarrow \genwatermark(\sktrace, \tid, b, h'); \\
        b' \leftarrow \CA^{\genwatermark(\sktrace,\,\cdot,\,\cdot,\,\cdot)}(w, \sigma); \\
        b' = b
      \end{array}
    \right] \leq \frac{1}{2} + \negl(\secparam).
  \]
\end{definition}
Ephemeral traceability is the central guarantee of the scheme: once a tag has been degraded beyond its hop budget, the authority irrevocably loses the ability to recover the traced identifier, even in possession of $\sktrace$. In the exponential ElGamal and Damg{\aa}rd--Jurik instantiations the guarantee is stronger still: degradation rerandomizes with fresh uniform randomness, so a tag whose contributions have all expired is distributed exactly as a dummy tag, even to the holder of $\sktrace$.
\begin{definition}[Ephemeral traceability]
  \label{def:ephemeral}
  An ECT scheme satisfies \emph{ephemeral traceability} if for every identifier $\tid$, every per-tag budget $h' \le h$, and every $k \geq h'$:
  \[
    \Pr\left[
      \begin{array}{l}
        (\pp, \sktrace) \leftarrow \setup(1^\secparam); \\
        (w^{(0)}, \sigma) \leftarrow \genwatermark(\sktrace, \tid, 1, h'); \\
        w^{(k)} \leftarrow \degrade^k(w^{(0)}); \\
        \trace(\sktrace, w^{(k)}) \neq \bot
      \end{array}
    \right] \leq \negl(\secparam),
  \]
  where $\degrade^k(w^{(0)})$ again denotes $k$ sequential applications of $\degrade$.
\end{definition}
Unforgeability rules out framing: an adversary cannot derive a tag with a valid history that traces to an identifier not tagged by the authority. The notion cannot be stated for bare ciphertexts, because public encryption allows anyone to encrypt arbitrary identifiers; it only concerns tags with a valid history. Let $R_{\mathsf{hist}}$ contain every $(\pp,w^*,H)$ such that $H$ derives $w^*$ from issued tags through the payment and ingestion transitions of Section~\ref{sec:tracing-scheme}. Every issuance in $H$ must satisfy $\mathsf{Verify}(w,\sigma)=1$. A payment degrades a sender tag $w_S$ to $\hat w_S$, reproducibly from its witness $\xi$, and outputs $\hat w_S$ and $\mergewatermarks(\hat w_S,w_R)$ for a recipient tag $w_R$. An ingestion degrades an account tag in the same way and merges one issued tag into the result. No tag may be consumed twice. A ledger proof with soundness and nullifiers is one way to ensure that every accepted tag has such a history; direct verification or an attested execution environment may provide the same guarantee.
\begin{definition}[Unforgeability]
  \label{def:unforge}
  An ECT scheme satisfies \emph{unforgeability} if for every PPT adversary $\CA$, letting $\CQ$ denote the set of identifiers queried to the $\genwatermark$ oracle with flag $b = 1$, where $\CA$ queries each identifier at most once with $b=1$:
  \[
    \Pr\left[
      \begin{array}{l}
        (\pp, \sktrace) \leftarrow \setup(1^\secparam); \\
        (w^*, H) \leftarrow \CA^{\genwatermark(\sktrace,\,\cdot,\,\cdot,\,\cdot)}(\pp); \\
        (\pp,w^*,H)\in R_{\mathsf{hist}} \;\wedge{} \\
        \trace(\sktrace, w^*) \neq \bot \;\wedge{} \\
        \qquad \exists\, \tid \in \trace(\sktrace, w^*) :\; \tid \notin \CQ
      \end{array}
    \right] \leq \negl(\secparam).
  \]
\end{definition}
The game allows one nonzero issuance per identifier because Lemma~\ref{lem:multiplicity} bounds the copies of a single contribution. This matches the tag lifecycle of Section~\ref{sec:system-model}, where a renewal issues a fresh identifier; reusing a slot after all its earlier contributions have expired is equivalent, since expired contributions are zero.
Unlinkability prevents any party without the tracing key from correlating a tag before and after a degradation step. The game below isolates the protection supplied by rerandomization; in a deployment the host ledger's zero-knowledge transaction proof additionally hides which committed state was consumed.
\begin{definition}[Unlinkability]
  \label{def:unlink}
  An ECT scheme satisfies \emph{unlinkability} if for every PPT adversary $\CA$:
  \[
    \Pr\left[
      \begin{array}{l}
        (\pp, \sktrace) \leftarrow \setup(1^\secparam); \\
        (w_0, w_1) \leftarrow \CA^{\genwatermark(\sktrace,\,\cdot,\,\cdot,\,\cdot)}(\pp); \\
        b \xleftarrow{\$} \{0,1\}; \\
        (w',\xi) \leftarrow \degrade(w_b); \\
        b' \leftarrow \CA(w_0, w_1, w'); \\
        b' = b
      \end{array}
    \right] \leq \frac{1}{2} + \negl(\secparam).
  \]
\end{definition}
Having defined the interface and the security properties of ECT, we now give three constructions that satisfy all six properties.
\section{Constructions}
\label{sec:algebraic-construction}
Standard properties of homomorphic encryption and signatures provide most of the required guarantees: additive homomorphism enables merging, IND-CPA security hides whether a tag is traced, rerandomization gives unlinkability, and EUF-CMA signatures authenticate issuance. Transition validation is handled separately by the deployment. The main challenge within ECT is \emph{ephemeral traceability}. After exactly $h$ steps, a tag must become indistinguishable from a dummy tag without destroying live contributions when tags are merged. We first examine two natural approaches that each satisfy only one side of this requirement.

\subsection{Approaches from first principles}\label{sec:strawman}

One approach is \emph{noisy rerandomization}. Each degradation step adds bounded noise $\delta \in [-\Delta,\Delta]$ to the encoded identifier. Once enough noise has accumulated, decoding fails and $\trace$ returns $\bot$ with overwhelming probability. Decoding failure, however, is not the same as expiry. Even after $h$ hops, the decryption of a traced tag remains concentrated around $m_i\cdot M$, whereas a dummy tag decrypts to zero. The authority can therefore still distinguish expired traced tags from dummy tags, and the tracing window does not close at a sharp threshold.

A second approach makes an expired tag look uniformly random. For example, pseudorandom-code constructions add noise from the ambient space until the codeword becomes computationally indistinguishable from uniform and its identifier can no longer be decoded, as in the construction of Cohen \etal~\cite{DBLP:conf/sp/Cohen0S25}. This gives a stronger notion of expiry, but it is incompatible with merging. A uniformly random expired tag is not a neutral element: merging it with a live tag randomizes the result and destroys the live identifier as well.

These approaches expose the two requirements that degradation must satisfy simultaneously. Expired tags must become indistinguishable from dummy tags, and they must behave neutrally when merged with tags that remain live. The nilpotent construction below achieves both by making every expired contribution equal to zero.

\subsection{Nilpotent degradation}
\label{sec:template}
The template requires an additively homomorphic encryption scheme whose plaintext ring $R$ contains a publicly applicable element $\nu$ that is nilpotent of index $h$: $\nu^h=0$ and $\nu^{h-1}\ne0$. A tag encrypts an identifier, degradation multiplies its plaintext by $\nu$ and rerandomizes the ciphertext, and merging adds plaintexts homomorphically. After exactly $h$ degradations, every contribution becomes zero and remains neutral under merging. Before then, the encoding keeps live contributions separable. Degradation requires no secret key; Section~\ref{app:discussion} further discusses which party may apply it.

We instantiate the template with exponential ElGamal over $\BZ_{q^h}$, using $\nu=q$ and $c\mapsto c^q$ (Section~\ref{sec:elgamal}); Damg{\aa}rd--Jurik over $\BZ_{N^h}$, using $\nu=N$ and $c\mapsto c^N$ (Section~\ref{sec:dj}); and Ring-LWE over $R_0[Y]/(Y^h)$, using $\nu=Y$ (Section~\ref{sec:RLWE}). ElGamal gives the smallest tags, Damg{\aa}rd--Jurik gives polynomial-time tracing under DCR, and Ring-LWE provides post-quantum security using Chinese remainder theorem (CRT) slots instead of positional encoding.
\subsection{Exponential ElGamal instantiation}
\label{sec:elgamal}
Exponential ElGamal gives compact tags and is designed for tags containing a small set of distinguishable live identifiers, since tracing becomes expensive as this set grows. The public parameters bound the identifier-slot capacity, but bounding the number of tagged accounts requires a separate issuance policy. The construction uses a cyclic group $\BG$ of order $q^h$ for prime $q$; this prime-power order requires a non-standard DDH assumption and new parameters for each hop budget.
\begin{definition}[Prime-power-order group]
  A \emph{cyclic group of prime-power order} $q^h$ is a group $\BG=\langle g\rangle$ in which $g^{q^h}=1$ and $g^k\ne1$ for $0<k<q^h$.
\end{definition}
Exponentiation by $q$ realizes nilpotent degradation: applying $x\mapsto x^q$ exactly $h$ times gives $x^{q^h}=1$. Thus, $q$ is the nilpotent element of Section~\ref{sec:template}.
Gennaro, Krawczyk, and Rabin reduce DDH in cyclic groups to their maximal prime-power-order components, but not further to prime-order subgroups~\cite{DBLP:conf/eurocrypt/GennaroKR04}. We therefore state the required assumption directly.
\begin{definition}[DDH in prime-power-order groups]
  \label{def:ddh-prime-power}
  Let $\BG$ be a cyclic group of order $q^h$ with generator $g$. The \emph{Decisional Diffie-Hellman (DDH) assumption} holds in $\BG$ if for all PPT distinguishers $\CD$, where $a, b, c \xleftarrow{\$} \BZ_{q^h}$,
  \[
    \left|\Pr[\CD(g,\, g^a,\, g^b,\, g^{ab}) = 1] - \Pr[\CD(g,\, g^a,\, g^b,\, g^c) = 1]\right| \leq \negl(\secparam).
  \]
\end{definition}
Projection by $x\mapsto x^{q^{h-1}}$ shows that prime-order DDH is necessary, but the converse is not known. Boneh's random self-reduction gives worst-case-to-average-case equivalence when $q$ is super-polynomial~\cite{DBLP:conf/ants/Boneh98}, while Shoup bounds generic advantage after $m$ operations by $O(m^2/q)$~\cite{DBLP:conf/eurocrypt/Shoup97}. Hence, generic attacks require $\Omega(\sqrt q)$ operations, matching Pohlig--Hellman's $O(h\sqrt q)$ cost up to $h$~\cite{DBLP:journals/tit/PohligH78}. These results support, but do not prove, our assumption. One realization is the order-$q^h$ subgroup of $\BZ_p^*$ for a prime $p$ such that $q^h\mid p-1$.

With the algebraic setting and security assumption fixed, we now present the six algorithms of the exponential-ElGamal construction in lifecycle order. We begin with $\setup$ (Alg.~\ref{alg:eg-setup}), which generates an ElGamal key pair for tag encryption and a signature key pair for authenticating issued tags.
\begin{algorithm}[ht!]
\caption{Setup.}
\label{alg:eg-setup}
\Alg{$\setup(1^\secparam, n, h, r) \rightarrow (\pp, \sktrace)$}{
  Let $r$ be an integer with $r > 2^{h-1}$ \tcp*{encoding base, Lemma~\ref{lem:multiplicity}}
  Choose prime $q$ with $q > r^n$ and $q \geq 2^{2\secparam}$, $q \neq r$\;
  Construct cyclic group $\BG$ of order $q^h$ with generator $g$\;
  $(\mathsf{pk}_{\text{enc}}, \mathsf{sk}_{\text{enc}}) \leftarrow \mathsf{KeyGen}(\BG)$\;
  $(\mathsf{pk}_{\text{sig}}, \mathsf{sk}_{\text{sig}}) \leftarrow \mathsf{SigKeyGen}(1^\secparam)$\;
  Assign $m_i \gets r^{i-1}$ for each user $\tid_i$, $i \in [n]$\;
  $\pp \gets (\BG, q, h, n, r, \mathsf{pk}_{\text{enc}}, \mathsf{pk}_{\text{sig}})$, \quad $\sktrace \gets (\mathsf{sk}_{\text{enc}}, \mathsf{sk}_{\text{sig}})$\;
  \KwRet{$(\pp, \sktrace)$}\;
}
\end{algorithm}
\FloatBarrier
The public parameters make both bounds auditable. For $\BG\subseteq\BZ_p^*$, the conditions $q^h\mid p-1$, $g^{q^h}=1$, and $g^{q^{h-1}}\ne1$ establish the hop budget $h$. They also give capacity $n_{\max}=\lfloor\log_r q\rfloor$; larger encodings cross a base-$q$ digit and require running \setup with new parameters.

$\genwatermark$ (Alg.~\ref{alg:eg-tag}) realizes budget $h'\le h$ by issuing at depth $h-h'$. Setting $b=0$ encrypts zero. The authority signs each issued ciphertext, and $\mathsf{Verify}$ (Alg.~\ref{alg:eg-verify}) checks the signature under $\mathsf{pk}_{\text{sig}}$.
\begin{figure}[ht!]
\centering
\begin{minipage}[t]{0.49\textwidth}
\begin{algorithm}[H]
\caption{Tag issuance.}
\label{alg:eg-tag}
\Alg{$\genwatermark(\sktrace, \tid_i, b, h') \rightarrow (w, \sigma)$}{
  $m \gets q^{h-h'} \cdot r^{i-1} \cdot b$ \;
  $w \gets \mathsf{Enc}(\mathsf{pk}_{\text{enc}}, m)$\;
  $\sigma \gets \mathsf{Sign}(\mathsf{sk}_{\text{sig}}, w)$\;
  \KwRet{$(w, \sigma)$}\;
}
\end{algorithm}
\end{minipage}\hfill%
\begin{minipage}[t]{0.49\textwidth}
\begin{algorithm}[H]
\caption{Verification.}
\label{alg:eg-verify}
\Alg{$\mathsf{Verify}(w,\sigma) \rightarrow \{0,1\}$}{
  \KwRet{$\mathsf{Ver}(\mathsf{pk}_{\text{sig}},w,\sigma)$}\;
}
\end{algorithm}
\end{minipage}
\end{figure}
\FloatBarrier
$\degrade$ (Alg.~\ref{alg:eg-degrade}) raises both components to $q$ and rerandomizes, implementing $m\mapsto qm\bmod q^h$. It returns the rerandomization exponent as $\xi$.
The merging algorithm (Alg.~\ref{alg:eg-merge}) multiplies ciphertexts componentwise, thereby adding their plaintexts. The transaction statement enforces correctness.
\begin{figure}[ht!]
\centering
\begin{minipage}[t]{0.49\textwidth}
\begin{algorithm}[H]
\caption{Degradation.}
\label{alg:eg-degrade}
\Alg{$\degrade(w) \rightarrow (w',\xi)$}{
  Parse $w = (A, B)$\;
  $r' \xleftarrow{\$} \BZ_{q^h}$\;
  $w' \gets \bigl(A^q \cdot g^{r'},\, B^q \cdot \mathsf{pk}_{\text{enc}}^{r'}\bigr)$\;
  \KwRet{$(w',\xi=r')$}\;
}
\end{algorithm}
\end{minipage}\hfill%
\begin{minipage}[t]{0.49\textwidth}
\begin{algorithm}[H]
\caption{Merging.}
\label{alg:eg-merge}
\Alg{$\mergewatermarks(w_1, \ldots, w_m) \rightarrow w'$}{
  $w' \gets \prod_{j=1}^m w_j$\;
  \KwRet{$w'$}\;
}
\end{algorithm}
\end{minipage}
\end{figure}
\FloatBarrier
$\trace$ (Alg.~\ref{alg:eg-trace}) recovers $v\in\BZ_{q^h}$ from $g^v$: base $q$ separates degradation depths, and base $r$ separates identifiers within each depth. Let $I\subseteq[n]$ contain the $K$ identifier slots issued so far, which the authority records with \sktrace, split as $I=I_0\sqcup I_1$. For $b\in\{0,1\}$, let $D_b$ contain all sums $\sum_{i\in I_b}c_i r^{i-1}$ with $0\le c_i\le2^{h-1}$, and set $D=D_0+D_1\subseteq\BZ_q$. Lemma~\ref{lem:multiplicity} gives the coefficient bound, and $|D|\le r^K$ because $r>2^{h-1}$.

As in Pohlig--Hellman~\cite{DBLP:journals/tit/PohligH78}, tracing extracts one digit at a time. With $\gamma=g^{q^{h-1}}$ and $y_t=g^{\sum_{j\ge t}d_jq^j}$, the equation $y_t^{q^{h-1-t}}=\gamma^{d_t}$ determines $d_t\in D$; multiplying by $g^{-d_tq^t}$ removes it. A meet-in-the-middle search over $D_0$ and $D_1$ finds each digit in $O(r^{K/2})$ time and memory. Overall, tracing costs $O(h r^{K/2})$ group operations and $O(r^{K/2})$ memory, so polynomial-time tracing requires $K=O(\log\secparam)$. Damg{\aa}rd--Jurik removes this restriction (Section~\ref{sec:dj}). If no valid digit exists, $\trace$ returns $\bot$.
\begin{algorithm}[ht!]
\caption{Tracing.}
\label{alg:eg-trace}
\Alg{$\trace(\sktrace, w) \rightarrow S$ or $\bot$}{
  Parse $w = (A, B)$ and set $y_0 \gets B / A^{\mathsf{sk}_{\text{enc}}}$ \tcp*{$y_0 = g^{v}$}
  \lIf{$y_0 = 1$}{\KwRet{$\bot$}}
  $\gamma \gets g^{q^{h-1}}$\;
  \For{$t = 0, \ldots, h-1$}{
    \lIf{no $d \in D$ satisfies $\gamma^{d} = y_t^{\,q^{h-1-t}}$}{\KwRet{$\bot$}}
    $d_t \gets$ the $d \in D$ with $\gamma^{d} = y_t^{\,q^{h-1-t}}$\;
    $y_{t+1} \gets y_t \cdot g^{-d_t q^{t}}$\;
  }
  Write $d_t = \sum_{i \in I} c_{t,i}\, r^{i-1}$ for each $t$\;
  \KwRet{$S \gets \{\tid_i : c_{t,i} > 0 \text{ for some } t\}$}\;
}
\end{algorithm}
\FloatBarrier
Under DDH in $\BG$ and EUF-CMA security of the issuance signature, the construction satisfies the six properties of Section~\ref{subsec:security-defs}, as we now show.
\label{app:elgamal-proofs}
\begin{theorem}[Trace correctness]
  \label{thm:trace-correctness}
  For any $\tid_i \in [n]$, any per-tag budget $h' \le h$, and any $t < h'$, let $w \leftarrow \genwatermark(\sktrace, \tid_i, 1, h')$ and $w'$ be obtained by applying $\degrade$ exactly $t$ times. Then $\tid_i \in \trace(\sktrace, w')$ with probability $1 - \negl(\secparam)$.
\end{theorem}
\begin{proof}
  By the ElGamal homomorphism, raising each ciphertext component to the $q$-th power multiplies the plaintext by $q$ in $\BZ_{q^h}$. The tag $\genwatermark(\sktrace, \tid_i, 1, h')$ encrypts $q^{\,h-h'} r^{i-1}$ (issuance depth $h-h'$); after $t$ applications of $\degrade$ it encrypts $q^{\,h-h'+t} \cdot r^{i-1} \bmod q^h$. Since $r^{i-1} < q$ (by $q > r^n$) and $h-h'+t < h$ for $t < h'$, this value is nonzero. Writing it in base $q$, the digit at position $h-h'+t$ is $r^{i-1}$ and all others are zero; decomposing $r^{i-1}$ in base $r$ gives a unit coefficient at position $i-1$, so $\trace$ returns $\tid_i$. \qed
\end{proof}
\begin{theorem}[Merge correctness]
  \label{thm:merge-correctness}
  For $i \in [m]$, let $w_i \leftarrow \genwatermark(\sktrace, \tid_i, b_i, h'_i)$, and suppose each $w_i$ has been degraded $t_i \geq 0$ times before merging; write $\tau_i = (h - h'_i) + t_i$ for its current depth and let $L = \{i : b_i = 1 \wedge \tau_i < h\}$ be the live inputs. Write $m(\tid_i)$ for the encoding assigned to the identifier of input $i$, allowing identifiers to repeat among inputs. Assume the multiplicity condition of Definition~\ref{def:merge-correctness}, and let $w' \leftarrow \mergewatermarks(w_1, \ldots, w_m)$. Then the merged ciphertext encrypts $\sum_{i \in L} q^{\tau_i} \cdot m(\tid_i)$, and $\Pr[\trace(\sktrace, w') = S_L] \ge 1 - \negl(\secparam)$, with $S_L$ as in that definition.
\end{theorem}
\begin{proof}
  The tag $\genwatermark(\sktrace, \tid_i, b_i, h'_i)$ is issued at depth $h - h'_i$, so after $t_i$ applications of $\degrade$ it encrypts $q^{\tau_i} \cdot m(\tid_i) \bmod q^h$ by the ElGamal homomorphism. For $\tau_i \geq h$ this value is zero modulo $q^h$, so expired inputs, like dummy inputs, contribute nothing. Merging via component-wise multiplication therefore encrypts $v = \sum_{i \in L} q^{\tau_i} \cdot m(\tid_i) \pmod{q^h}$. If $L = \emptyset$, then $v = 0$ and $\trace$ returns $\bot = S_L$. Otherwise, let $c_{\tau,j}$ count the live inputs with identifier $\tid_j$ at depth $\tau$. The partial sum at that depth is $d_\tau = \sum_{j=1}^n c_{\tau,j}r^{j-1}$. By the stated multiplicity condition, $c_{\tau,j}\le 2^{h-1}<r$, so $d_\tau\le\sum_{j=1}^n(r-1)r^{j-1}=r^n-1<q$. Hence no base-$q$ carry occurs across depths. The digit at position $\tau$ is exactly $d_\tau$, whose base-$r$ digits uniquely identify the live contributions at that depth. Taking the union over all $\tau$ yields $S_L$. \qed
\end{proof}
\begin{theorem}[Indistinguishability]
  \label{thm:indist}
  Under DDH in $\BG$, the ECT scheme satisfies indistinguishability.
\end{theorem}
\begin{proof}
  Let $m_1=q^{\,h-h'}r^{i-1}$ and $m_0=0$. On input $b$, $\genwatermark$ encrypts $m_b$ and signs the resulting ciphertext. Suppose $\CD$ distinguishes the two cases with advantage $\epsilon$. Construct an IND-CPA adversary $\CB$: on input $\mathsf{pk}_{\text{enc}}$, generate an independent signing key pair $(\mathsf{pk}_{\text{sig}}, \mathsf{sk}_{\text{sig}})$, answer tagging queries by encrypting under $\mathsf{pk}_{\text{enc}}$ and signing with $\mathsf{sk}_{\text{sig}}$, and submit $(m_1,m_0)$ to the IND-CPA challenger. Upon receiving the challenge ciphertext $w$, compute $\sigma \leftarrow \mathsf{Sign}(\mathsf{sk}_{\text{sig}}, w)$ and forward $(w, \sigma)$ to $\CD$. Since $\mathsf{sk}_{\text{sig}}$ is independent of the IND-CPA challenge, $(w, \sigma)$ has the same distribution as in the real game. Thus, $\CB$ achieves advantage $\epsilon$ against IND-CPA. By IND-CPA security of ElGamal under DDH in $\BG$, $\epsilon \leq \negl(\secparam)$. \qed
\end{proof}
\begin{theorem}[Ephemeral traceability]
  \label{thm:ephemeral}
  The ECT scheme satisfies ephemeral traceability unconditionally. Specifically, for any $\tid_i \in [n]$ and any per-tag budget $h' \le h$, if $(w^{(0)}, \sigma) \leftarrow \genwatermark(\sktrace, \tid_i, 1, h')$ and $w^{(k)}$ is obtained by applying $\degrade$ at least $h'$ times, then $\trace(\sktrace, w^{(k)}) = \bot$.
\end{theorem}
\begin{proof}
  The tag is issued at depth $h - h'$, so after $h'$ applications of $\degrade$ it encrypts $q^{h} \cdot m_i \equiv 0 \pmod{q^h}$, since the group has order $q^h$. As $0$ is a fixed point of $m \mapsto qm \bmod q^h$, all further degradations leave the plaintext at $0$. Decryption yields $0$, so $\trace$ returns $\bot$. The argument is information-theoretic; no computational assumption is used. \qed
\end{proof}
\begin{theorem}[Unforgeability]
  \label{thm:unforge}
  Under EUF-CMA security of the signature scheme, the exponential ElGamal instantiation satisfies unforgeability: every tag with a valid history traces only to identifiers issued by the authority. A deployment inherits unforgeability when its verification mechanism accepts only tags possessing such a history.
\end{theorem}
\begin{proof}
  Let $(w^*, H)$ be the adversary's output with $(\pp,w^*,H)\in R_{\mathsf{hist}}$. A reduction $\CB$ against EUF-CMA receives $\mathsf{pk}_{\text{sig}}$ and a signing oracle, generates the encryption key pair itself, and simulates the $\genwatermark$ oracle by encrypting the prescribed plaintext and querying the signing oracle. An issued tag in $H$ that the oracle never produced is a valid forgery, so, except with negligible probability, every issued tag in $H$ encrypts $q^{h-h'} r^{i-1}$ with $\tid_i \in \CQ$, or zero. We argue by induction along $H$ that every derived tag encrypts a value of the form $\sum_j \kappa_j q^{t_j}r^{i_j-1}$ with $\tid_{i_j}\in\CQ$: degradation multiplies by $q$, a payment copies the degraded sender contribution into at most two outputs, ingestion merges one issued contribution at its issuance depth, and no tag is consumed twice. Since each identifier is issued with $b=1$ at most once, it has multiplicity at most $2^t\le2^{h-1}$ at every live depth $t$. The choices $r>2^{h-1}$ and $q>r^n$ prevent all carries, so $\trace(\sktrace,w^*)$ returns only identifiers in $\CQ$. \qed
\end{proof}
\begin{theorem}[Unlinkability]
  \label{thm:unlink}
  Under the DDH assumption in $\BG$, the ECT scheme satisfies unlinkability. In a ledger deployment, zero knowledge of the host transaction proof additionally hides the consumed account states.
\end{theorem}
\begin{proof}
  Let $\CA$ be a PPT adversary submitting $(w_0, w_1)$ and receiving $w' \leftarrow \degrade(w_b)$ for a uniformly random $b$. A tag parses as a pair of elements of $\BG$, membership being publicly checkable via $x^{q^h} = 1$, so we may assume $w_0, w_1 \in \BG \times \BG$.
  Write $w_b = (A_b, B_b)$. Degradation returns $w' = \bigl(A_b^q \cdot g^{r'},\, B_b^q \cdot \mathsf{pk}_{\text{enc}}^{r'}\bigr)$ for uniform $r' \in \BZ_{q^h}$, that is, $w_b$ raised componentwise to the $q$-th power and multiplied by the rerandomization factor $(g^{r'}, \mathsf{pk}_{\text{enc}}^{r'})$. Construct a distinguisher $\CB$ for DDH: on challenge $(g, X, Y, Z)$, set $\mathsf{pk}_{\text{enc}} \gets X$, answer tagging queries by encrypting under $X$ and signing with a self-generated key, sample $b$, and give $\CA$ the challenge $w' = (A_b^q \cdot Y,\, B_b^q \cdot Z)$; output $1$ if $\CA$ guesses $b$. If $(g, X, Y, Z)$ is a Diffie--Hellman tuple, then $Y = g^{r'}$ and $Z = X^{r'}$ for uniform $r'$, and $w'$ is distributed exactly as $\degrade(w_b)$. If $Z$ is instead uniform, then $(Y, Z)$ is uniform on $\BG \times \BG$, so $w'$ is uniform and independent of $b$, and $\CA$ guesses $b$ with probability exactly $\frac{1}{2}$. The advantage of $\CA$ is therefore bounded by the DDH advantage of $\CB$, and $\Pr[b' = b] \leq \frac{1}{2} + \negl(\secparam)$. \qed
\end{proof}
\subsection{Damg{\aa}rd--Jurik instantiation}
\label{sec:dj}
\label{app:dj}
When a deployment requires more simultaneous traces than the ElGamal construction efficiently supports, Damg{\aa}rd--Jurik is a natural choice. It trades larger tags for a larger plaintext space, polynomial-time tracing, and security under DCR. Its plaintext ring is $\BZ_{N^s}$, where $N$ is an RSA modulus, and the ciphertext space is $\BZ_{N^{s+1}}^*$. The modulus itself plays the role of the nilpotent element: raising a ciphertext to the power $N$ multiplies its plaintext by $N$ modulo $N^s$,
\[
  c^{N} = (1+N)^{Nm} \, \rho^{N \cdot N^s} = (1+N)^{Nm} \, (\rho^{N})^{N^s} = \mathsf{Enc}\bigl(N,\, N m \bmod N^s;\ \rho^{N}\bigr),
\]
so setting $s=h$ gives exact expiry after $h$ hops. A tag is a single element of $\BZ_{N^{h+1}}^*$, degradation is $w \mapsto w^{N} \rho^{N^{h}} \bmod N^{h+1}$ for a fresh $\rho$, and merging is multiplication modulo $N^{h+1}$. Tracing decrypts to $v \in \BZ_{N^{h}}$, reads the base-$N$ digits as degradation depths and the base-$r$ digits within each as identifiers, and returns $\bot$ if the digit bound $d_t < r^n$ fails.

The plaintext space is a full residue ring rather than a discrete-log space, so tracing costs $O(nh)$ modular operations instead of a bounded search, and the capacity grows with $\log_r N$ rather than with $\log_r q$. The same six properties hold: correctness and ephemeral traceability by the same digit argument as above, indistinguishability and unlinkability under DCR in place of DDH, and unforgeability again under EUF-CMA alone. We now give the assumption, the six algorithms, and the proofs.
Identifiers use $m_i=r^{i-1}$ with $r>2^{h-1}$ and $N>r^n$, so base $N$ separates degradation depths and base $r$ separates identifiers within each depth.
\begin{definition}[Decisional composite residuosity]
  \label{def:dcr}
  Let $N = p_1 p_2$ be an RSA modulus and $s \geq 1$. The \emph{decisional composite residuosity (DCR) assumption} holds if for all PPT distinguishers $\CD$:
  \[
    \left|\Pr[\CD(N, z) = 1] - \Pr[\CD(N, \rho^{N^s} \bmod N^{s+1}) = 1]\right| \leq \negl(\secparam)
  \]
  where $z \xleftarrow{\$} \BZ_{N^{s+1}}^*$ and $\rho \xleftarrow{\$} \BZ_{N^{s+1}}^*$, that is, a uniform element of $\BZ_{N^{s+1}}^*$ cannot be distinguished from a uniform $N^s$-th residue.
\end{definition}
Damg{\aa}rd--Jurik is IND-CPA secure under DCR for every $s$, with the assumption equivalent to the $s=1$ case~\cite{DBLP:conf/pkc/DamgardJ01}. We again use $m_i=r^{i-1}$, choose $r>2^{h-1}$ and $N>r^n$, and let base $N$ separate depths. $\setup$ (Alg.~\ref{alg:dj-setup}) samples the modulus and fixes $s = h$.
\begin{algorithm}[ht!]
\caption{Setup.}
\label{alg:dj-setup}
\Alg{$\setup(1^\secparam, n, h, r) \rightarrow (\pp, \sktrace)$}{
  Let $r$ be an integer with $r > 2^{h-1}$ \tcp*{prevents multiplicity carries}
  Sample distinct primes $p_1, p_2$ and set $N \gets p_1 p_2$, with $N > r^n$ and $|N|$ sized for $\secparam$-bit factoring security\;
  $s \gets h$ \tcp*{plaintext ring $\BZ_{N^{h}}$, ciphertext ring $\BZ_{N^{h+1}}^*$}
  $\mathsf{pk}_{\text{enc}} \gets N$, \quad $\mathsf{sk}_{\text{enc}} \gets \lambda_N = \mathrm{lcm}(p_1 - 1, p_2 - 1)$\;
  $(\mathsf{pk}_{\text{sig}}, \mathsf{sk}_{\text{sig}}) \leftarrow \mathsf{SigKeyGen}(1^\secparam)$\;
  Assign $m_i \gets r^{i-1}$ for each user $\tid_i$, $i \in [n]$\;
  $\pp \gets (N, h, n, r, \mathsf{pk}_{\text{sig}})$, \quad $\sktrace \gets (\mathsf{sk}_{\text{enc}}, \mathsf{sk}_{\text{sig}})$\;
  \KwRet{$(\pp, \sktrace)$}\;
}
\end{algorithm}
\FloatBarrier
The public tag space $\BZ_{N^{h+1}}$ exposes the hop budget, while $n_{\max}=\lfloor\log_r N\rfloor$ exposes the identifier-slot capacity. Enlarging either bound enlarges every tag.

$\genwatermark$ (Alg.~\ref{alg:dj-tag}) issues at depth $h-h'$, encrypts zero for a dummy, and signs the resulting ciphertext, which $\mathsf{Verify}$ (Alg.~\ref{alg:dj-verify}) checks against the authority's issuance key.
\begin{figure}[ht!]
\centering
\begin{minipage}[t]{0.49\textwidth}
\begin{algorithm}[H]
\caption{Tag issuance.}
\label{alg:dj-tag}
\Alg{$\genwatermark(\sktrace, \tid_i, b, h') \rightarrow (w, \sigma)$}{
  $m \gets N^{h-h'} \cdot r^{i-1} \cdot b$\;
  $\rho \xleftarrow{\$} \BZ_{N^{h+1}}^*$\;
  $w \gets (1+N)^{m} \rho^{N^{h}} \bmod N^{h+1}$\;
  $\sigma \gets \mathsf{Sign}(\mathsf{sk}_{\text{sig}}, w)$\;
  \KwRet{$(w, \sigma)$}\;
}
\end{algorithm}
\end{minipage}\hfill%
\begin{minipage}[t]{0.49\textwidth}
\begin{algorithm}[H]
\caption{Verification.}
\label{alg:dj-verify}
\Alg{$\mathsf{Verify}(w,\sigma) \rightarrow \{0,1\}$}{
  \KwRet{$\mathsf{Ver}(\mathsf{pk}_{\text{sig}},w,\sigma)$}\;
}
\end{algorithm}
\end{minipage}
\end{figure}
\FloatBarrier
$\degrade$ (Alg.~\ref{alg:dj-degrade}) implements $m\mapsto Nm\bmod N^h$ and returns the rerandomization value $\rho$ as $\xi$. $\mergewatermarks$ (Alg.~\ref{alg:dj-merge}) multiplies ciphertexts, adding their plaintexts modulo $N^h$.
\begin{figure}[ht!]
\centering
\begin{minipage}[t]{0.49\textwidth}
\begin{algorithm}[H]
\caption{Degradation.}
\label{alg:dj-degrade}
\Alg{$\degrade(w) \rightarrow (w',\xi)$}{
  $\rho \xleftarrow{\$} \BZ_{N^{h+1}}^*$\;
  $w' \gets w^{N} \cdot \rho^{N^{h}} \bmod N^{h+1}$\;
  \KwRet{$(w',\xi=\rho)$}\;
}
\end{algorithm}
\end{minipage}\hfill%
\begin{minipage}[t]{0.49\textwidth}
\begin{algorithm}[H]
\caption{Merging.}
\label{alg:dj-merge}
\Alg{$\mergewatermarks(w_1, \ldots, w_m) \rightarrow w'$}{
  $w' \gets \prod_{j=1}^m w_j \bmod N^{h+1}$\;
  \KwRet{$w'$}\;
}
\end{algorithm}
\end{minipage}
\end{figure}
\FloatBarrier
In the ledger, the transaction statement enforces $w'=w^N\rho^{N^h}\bmod N^{h+1}$ for each degradation and enforces multiplication for each merge.
\Needspace{0.30\textheight}
$\trace$ (Alg.~\ref{alg:dj-trace}) decrypts $v$ directly, then reads base-$N$ depths and base-$r$ identifiers in $O(nh)$ digit operations. The value $v=0$ denotes either a dummy or an expired tag. A digit $d_t \geq r^n$ cannot arise from a valid transition history, and $\trace$ then returns $\bot$.
\begin{algorithm}[ht!]
\caption{Tracing.}
\label{alg:dj-trace}
\Alg{$\trace(\sktrace, w) \rightarrow S$ or $\bot$}{
  $v \leftarrow \mathsf{Dec}(\mathsf{sk}_{\text{enc}}, w) \in \BZ_{N^{h}}$\;
  \lIf{$v = 0$}{\KwRet{$\bot$}}
  Write $v = \sum_{t=0}^{h-1} d_t \cdot N^t$ in base $N$\;
  \lIf{$d_t \geq r^n$ for some $t$}{\KwRet{$\bot$}}
  \ForEach{$t$ with $d_t \neq 0$}{
    Decompose $d_t = \sum_{i=1}^{n} c_{t,i} \cdot r^{i-1}$ in base $r$\;
    $S_t \gets \{\tid_i : i \in [n],\, c_{t,i} > 0\}$\;
  }
  \KwRet{$\bigcup_t S_t$}\;
}
\end{algorithm}
\FloatBarrier
Unforgeability uses EUF-CMA signatures; the other security proofs use DCR.
\begin{theorem}[Trace correctness, Damg{\aa}rd--Jurik]
  \label{thm:dj-trace}
  For any $\tid_i \in [n]$, any per-tag budget $h' \le h$, and any $t < h'$, let $(w, \sigma) \leftarrow \genwatermark(\sktrace, \tid_i, 1, h')$ and let $w'$ be obtained by applying $\degrade$ exactly $t$ times. Then $\tid_i \in \trace(\sktrace, w')$ with probability $1 - \negl(\secparam)$.
\end{theorem}
\begin{proof}
  Raising a ciphertext to $N$ multiplies the plaintext by $N$ in $\BZ_{N^{h}}$, so the tag, which encrypts $N^{\,h-h'} r^{i-1}$, encrypts $N^{\,h-h'+t} r^{i-1} \bmod N^{h}$ after $t$ degradations. Since $r^{i-1} \leq r^{n-1} < N$ by the constraint $N > r^n$, and $h-h'+t < h$, the integer $N^{\,h-h'+t} r^{i-1}$ is positive and strictly below $N^{\,h-h'+t+1} \leq N^{h}$, hence nonzero in $\BZ_{N^{h}}$. Its base-$N$ digit at position $h-h'+t$ is $r^{i-1}$ and all other digits vanish; decomposing that digit in base $r$ gives a unit coefficient at position $i-1$, so $\trace$ returns $\tid_i$. \qed
\end{proof}
\begin{theorem}[Merge correctness, Damg{\aa}rd--Jurik]
  \label{thm:dj-merge}
  For $i \in [m]$, let $w_i \leftarrow \genwatermark(\sktrace, \tid_i, b_i, h'_i)$, and suppose each $w_i$ has been degraded $t_i \geq 0$ times before merging; write $\tau_i = (h - h'_i) + t_i$ for its current depth and let $L = \{i : b_i = 1 \wedge \tau_i < h\}$ be the live inputs. Write $m(\tid_i)$ for the encoding assigned to the identifier of input $i$, allowing identifiers to repeat among inputs. Assume the multiplicity condition of Definition~\ref{def:merge-correctness}, and let $w' \leftarrow \mergewatermarks(w_1, \ldots, w_m)$. Then $w'$ encrypts $\sum_{i \in L} N^{\tau_i} m(\tid_i) \bmod N^{h}$, and $$\Pr[\trace(\sktrace, w') = S_L] \ge 1 - \negl(\secparam),$$ with $S_L$ as in that definition.
\end{theorem}
\begin{proof}
  Componentwise multiplication adds plaintexts modulo $N^{h}$. An input with $\tau_i \geq h$ contributes zero modulo $N^h$, so expired inputs, like dummy inputs, vanish, which gives the stated value $v$. If $L = \emptyset$, then $v = 0$ and $\trace$ returns $\bot = S_L$. Otherwise, let $c_{\tau,j}$ count the live inputs with identifier $\tid_j$ at depth $\tau$. Their contribution at that depth is $d_\tau=\sum_{j=1}^n c_{\tau,j}r^{j-1}$. The multiplicity condition gives $c_{\tau,j}\le2^{h-1}<r$, and therefore $d_\tau\le\sum_{j=1}^n(r-1)r^{j-1}=r^n-1<N$. Every depth occupies one base-$N$ digit with no carry, and the base-$r$ digits uniquely identify its contributions. Taking the union over $\tau$ yields $S_L$. \qed
\end{proof}
\begin{theorem}[Indistinguishability, Damg{\aa}rd--Jurik]
  \label{thm:dj-indist}
  Under DCR (Definition~\ref{def:dcr}), the Damg{\aa}rd--Jurik instantiation satisfies indistinguishability.
\end{theorem}
\begin{proof}
  A tag issued with $b = 1$ is $\mathsf{Enc}(N, N^{\,h-h'} r^{i-1})$ and a tag issued with $b = 0$ is $\mathsf{Enc}(N, 0)$, so a distinguisher between them yields an IND-CPA adversary. On input the public key $N$, the adversary $\CB$ generates an independent signing key pair, answers tagging queries by encrypting under $N$ and signing with $\mathsf{sk}_{\text{sig}}$, submits the challenge messages $(N^{\,h-h'} r^{i-1},\, 0)$, signs the challenge ciphertext $w$ with $\mathsf{sk}_{\text{sig}}$, and forwards $(w, \sigma)$ to the distinguisher. The signing key is independent of the challenge, so the simulated pair is distributed as in the real game and $\CB$ inherits the distinguisher's advantage. Damg{\aa}rd--Jurik is IND-CPA secure under DCR~\cite{DBLP:conf/pkc/DamgardJ01}, so that advantage is negligible. \qed
\end{proof}
\begin{theorem}[Ephemeral traceability, Damg{\aa}rd--Jurik]
  \label{thm:dj-ephemeral}
  For any $\tid_i \in [n]$ and any per-tag budget $h' \le h$, if $(w^{(0)}, \sigma) \leftarrow \genwatermark(\sktrace, \tid_i, 1, h')$ and $w^{(k)}$ is obtained by applying $\degrade$ at least $h'$ times, then $\trace(\sktrace, w^{(k)}) = \bot$. The statement holds unconditionally and for every modulus $N$.
\end{theorem}
\begin{proof}
  The tag is issued at depth $h - h'$, so after $h'$ degradations it encrypts $N^{h} m_i \equiv 0 \pmod{N^{h}}$, and $0$ is a fixed point of $m \mapsto N m \bmod N^{h}$, so every further degradation leaves the plaintext at $0$ and $\trace$ returns $\bot$. The identity $N^{h} \equiv 0 \pmod{N^{h}}$ holds for every integer $N$, so the argument is information-theoretic and uses no property of the public parameters beyond the value of $h$ fixed by the ciphertext ring. \qed
\end{proof}
\begin{theorem}[Unforgeability, Damg{\aa}rd--Jurik]
  \label{thm:dj-unforge}
  Under EUF-CMA security of the signature scheme, the Damg{\aa}rd--Jurik instantiation satisfies unforgeability: every tag with a valid history traces only to identifiers issued by the authority. A deployment inherits unforgeability when its verification mechanism accepts only tags possessing such a history.
\end{theorem}
\begin{proof}
  A reduction against EUF-CMA simulates the $\genwatermark$ oracle with the signing oracle and its own encryption keys. Thus, for every $(\pp,w^*,H)\in R_{\mathsf{hist}}$, each issued tag in $H$ was, except with negligible probability, produced by the oracle and encrypts $N^{h-h'}r^{i-1}$ with $\tid_i\in\CQ$, or zero. Induction along $H$ shows that $w^*$ encrypts a sum $\sum_j\kappa_jN^{t_j}r^{i_j-1}$ with every $\tid_{i_j}\in\CQ$: degradation multiplies by $N$, a payment copies the degraded sender contribution into at most two outputs, ingestion merges one issued contribution at its issuance depth, and no tag is consumed twice. Since each identifier is issued with $b=1$ at most once, it has multiplicity at most $2^{h-1}<r$ at every live depth. Together with $N>r^n$, this prevents all digit carries, so tracing returns only identifiers in $\CQ$. \qed
\end{proof}
\begin{theorem}[Unlinkability, Damg{\aa}rd--Jurik]
  \label{thm:dj-unlink}
  Under DCR (Definition~\ref{def:dcr}), the Damg{\aa}rd--Jurik instantiation satisfies unlinkability. In a ledger deployment, zero knowledge of the host transaction proof additionally hides the consumed account states.
\end{theorem}
\begin{proof}
  The adversary submits $(w_0, w_1)$ and receives $w' \leftarrow \degrade(w_b)$ for a uniform bit $b$. A tag parses as an element of $\BZ_{N^{h+1}}^*$, membership being publicly checkable via a gcd with $N$, so we may assume $w_0, w_1 \in \BZ_{N^{h+1}}^*$.
  Degradation returns $w' = w_b^{N} \rho^{N^{h}} \bmod N^{h+1}$ for uniform $\rho \in \BZ_{N^{h+1}}^*$. Construct a distinguisher $\CB$ for DCR with $s = h$: on challenge $(N, z)$, answer tagging queries by encrypting under $N$ and signing with a self-generated key, sample $b$, give $\CA$ the challenge $w' = w_b^{N} \cdot z \bmod N^{h+1}$, and output $1$ if $\CA$ guesses $b$. If $z = \rho^{N^{h}}$ is a random $N^{h}$-th residue, $w'$ is distributed exactly as $\degrade(w_b)$. If $z$ is uniform in $\BZ_{N^{h+1}}^*$, then $w' = w_b^{N} \cdot z$ is uniform in $\BZ_{N^{h+1}}^*$ and independent of $b$, since multiplication by the fixed unit $w_b^{N}$ permutes the group; $\CA$ then guesses $b$ with probability exactly $\frac{1}{2}$. The advantage of $\CA$ is therefore bounded by the DCR advantage of $\CB$, and $\Pr[b' = b] \leq \frac{1}{2} + \negl(\secparam)$. \qed
\end{proof}
\subsection{Ring-LWE instantiation}
\label{sec:RLWE}
\label{app:rlwe}
Ring-LWE is the post-quantum alternative and, like Damg{\aa}rd--Jurik, supports more distinguishable live identifiers per tag than the ElGamal construction. Its larger identifier-slot capacity and post-quantum security come at the cost of substantially larger tags. The construction works over
\[
  R \;=\; \BZ[X,Y]\big/\bigl(X^{N_0}+1,\; Y^{h}\bigr) \;=\; R_0[Y]/(Y^{h}),
  \qquad R_0 = \BZ[X]/(X^{N_0}+1),
\]
where the second variable supplies the nilpotent element outright: $Y^{h}=0$ holds in $R$ by construction, so the hop budget is baked into the public ring rather than into the order of a group. Tags are BGV ciphertexts over $R$, degradation is $(c_0,c_1) \mapsto (Y c_0,\, Y c_1) + \mathsf{Enc}(\mathsf{pk}_{\text{enc}}, 0)$, and merging is ciphertext addition. Only additive homomorphism and multiplication by the public constant $Y$ are ever used, so plain IND-CPA BGV suffices, with neither key switching nor modulus switching.

Two features distinguish this instantiation from the previous two. First, choosing the plaintext modulus such that $t \equiv 1 \pmod{2N_0}$ splits the plaintext ring into $N_0$ CRT slots, one per identifier. Thus, the positional base-$r$ encoding of Definition~\ref{def:base-r-encoding} is unnecessary: the capacity is $n_{\max}=N_0$, and the only remaining constraint is that no slot wraps, namely $t > 2^{h-1}$. Second, nilpotency erases noise together with messages. After degradation, the invariant noise is $Y \cdot e_{\text{old}} + e_{\text{fresh}}$. Noise older than $h$ hops therefore vanishes, and a degradation chain of any length contains at most $h$ live rerandomization terms. Hardness is stated as decision RLWE over $R$, which is not cyclotomic but decomposes over the cyclotomic ring $R_0$. We now give the ring, the assumption, the algorithms, and the proofs.
Our Ring-LWE instantiation works over the nonstandard ring
\[
 R:=\frac{\mathbb{Z}[X]}{(X^{N_0}+1)}\otimes_{\mathbb Z}
       \frac{\mathbb{Z}[Y]}{(Y^h)}
 \;\cong\;\mathbb{Z}[X,Y]/(X^{N_0}+1,Y^h).
\]
We write $R_0:=\mathbb{Z}[X]/(X^{N_0}+1)$, so that
$R=R_0[Y]/(Y^h)$.  Every element has a unique representation
$\sum_{i=0}^{h-1}c_iY^i$ with $c_i\in R_0$.  Multiplication by $Y$ shifts
these coefficients and maps every plaintext to zero after $h$ applications.
Identifiers are encoded in CRT slots of the $X$ component, rather than by the
base-$r$ positional encoding used in the two preceding instantiations.

We use the coefficient $\ell_\infty$ norm on $R$, written
$\lVert a\rVert_\infty$ or simply $\lVert a\rVert$.  Let $\delta_R$ denote
an expansion bound satisfying
$\lVert ab\rVert_\infty\le\delta_R\lVert a\rVert_\infty
\lVert b\rVert_\infty$ for the products considered below.

As an $R_0$-module, $R$ is free of rank $h$ with basis $1,Y,\dots,Y^{h-1}$, so every element can be written uniquely as $u=\sum_{j<h}u_jY^j$. In the following we will refer to the coefficients $u_j\in R_0$ as the \emph{blocks} of $u$.

We use two elementary facts about $R$.
\begin{lemma}[Nilpotent shift]\label{lem:nilp}
Let $\mu_Y:R\to R$ be the map $m\mapsto Ym$. For every nonzero
$m_0\in R_0$ and $0\le d<h$, the value
$\mu_Y^d(m_0)=m_0Y^d$ is nonzero and lies in block $d$. Moreover,
$\mu_Y^h$ is the zero map and $\mu_Y(0)=0$. Thus $\mu_Y$ is nilpotent of
index exactly $h$, and zero is absorbing.
\end{lemma}
\begin{proof}
The blocks of an element of $Y$-degree $< h$ are its canonical coordinates, so $m_0 Y^d \ne 0$ for $d < h$; and $m_0 Y^h = m_0 \cdot 0 = 0$ since $Y^h = 0$ in $R$. The final claim, that zero is absorbing, is immediate.
\end{proof}
\begin{lemma}[Degradation]\label{lem:noisefree}
In $R = R_0[Y]/(Y^h)$, for every $e \in R$ and $d \ge 0$: $\lVert {Y^d e} \rVert \le \lVert e \rVert$.
\end{lemma}
\begin{proof}
Consider the case for $d=1$. Write $e = \sum_{j=0}^{h-1} e_j Y^j$ with blocks $e_j \in R_0$. Then $Y e = \sum_{j=0}^{h-2} e_j Y^{j+1}$: the blocks shift up by one and the top block goes to 0 (by \Cref{lem:nilp}). Hence $\lVert Ye \rVert = \max_{j \le h-2} \lVert e_j \rVert \le \lVert e \rVert$. The case for general $d$ follows by induction.
\end{proof}
Next we define the Ring-LWE problem over rings of the above kind.
\begin{definition}[Decision Ring-LWE over a ring]
\label{def:rlwe}
Let $R$ be a commutative ring and $Q$ a positive integer such that
$R_Q=R/QR$ is finite. Let $\chi_s,\chi_e$ be
distributions on $R$. The \emph{decision RLWE assumption} $\mathrm{RLWE}_{R, Q, \chi_s, \chi_e}$ holds if for every PPT distinguisher $\mathcal{D}$
and every $m=\mathrm{poly}(\lambda)$,
\[
  \Bigl|
  \Pr\bigl[\mathcal{D}\bigl(\{(a_i,\;a_is+e_i)\}_{i\in[m]}\bigr)=1\bigr]
  -\Pr\bigl[\mathcal{D}\bigl(\{(a_i,\;u_i)\}_{i\in[m]}\bigr)=1\bigr]
  \Bigr| \le \mathsf{negl}(\lambda),
\]
where $s\leftarrow\chi_s$ is drawn once and
$a_i,u_i\xleftarrow{\$}R_Q$ and $e_i\leftarrow\chi_e$ are drawn independently for each $i$.
\end{definition}
\begin{remark}
  Standard Ring-LWE is typically defined over a cyclotomic ring, where its hardness is well studied. Although $R$ is not cyclotomic, it is built over the cyclotomic ring $R_0$. Section~\ref{sec:rlwe-hardness} uses this structure to relate the assumption over $R$ to standard Ring-LWE over $R_0$.
\end{remark}
\begin{remark}
  Classical hardness reductions use a uniform secret and a suitably narrow discrete Gaussian error; standard transformations also support short secrets. Efficient implementations often use a uniform ternary secret together with discrete Gaussian error. Our proofs require only the stated Ring-LWE assumption and blockwise independence, while the concrete estimates follow the ternary-secret parameterization of~\cite{DBLP:journals/iacr/AlbrechtCCDGGHH19}.
\end{remark}
\paragraph{Related variants in the literature.}
Our assumption is closest to multivariate Ring-LWE, introduced and developed by Pedrouzo-Ulloa, Troncoso-Pastoriza, and P\'erez-Gonz\'alez~\cite{DBLP:conf/icassp/Pedrouzo-UlloaT15,DBLP:journals/corr/Pedrouzo-UlloaT16a,DBLP:journals/corr/abs-1712-00848,PTGGP21}. They consider rings of the form $\mathbb{Z}_q[x,y]/(f(x),g(y))$, where $f(x)=x^{n_1}+1$ and $g(y)=y^{n_2}+1$ are power-of-two cyclotomic polynomials. This structure was claimed to provide hardness equivalent to RLWE over $\mathbb{Z}_q[z]/(z^{n_1n_2}+1)$. Bootland, Castryck, and Vercauteren showed that this claim is false when $n_2$ divides $n_1$~\cite{BCV20}: evaluating $y$ at suitable roots in $\mathbb{Z}_q[x]/(x^{n_1}+1)$ maps multivariate samples to native dimension-$n_1$ RLWE samples with only $O(\sqrt{n_2})$ error growth.

Our second component $\mathbb{Z}[Y]/(Y^h)$ is instead non-cyclotomic and has only the repeated root $0$. The corresponding evaluation map is the projection $Y\mapsto0$, which exposes the cyclotomic base component but not the remaining $Y$-blocks. This explains both the connection to standard Ring-LWE and why the multivariate attack above does not directly recover the full secret in our ring.
\paragraph{The encryption scheme.}
For this instantiation, we use the BGV encryption scheme~\cite{DBLP:conf/innovations/BrakerskiGV12}, which supports fully homomorphic encryption. We need only additive homomorphism and multiplication by a fixed public ring element, so we describe the simplified scheme below.

The scheme is defined over a ring $R$ with moduli $Q$ and a parameter $t$ satisfying $\gcd(t,Q)=1$, with $\chi_s$ and $\chi_e$ as above. Here $R_Q$ serves as the ciphertext ring, while $P=R/tR$ is the plaintext ring. We then have:
\begin{algorithm}[ht!]
\caption{BGV, simplified.}
\label{alg:BGV}
$\mathsf{KeyGen}(R,Q,t)$:
$s \leftarrow \chi_s$ (ternary); \quad $a \leftarrow R_Q$ uniform; \quad $e \leftarrow \chi_e$
$\mathsf{pk} \leftarrow (p_0, p_1) := (-(a s + t e),\; a) \in R_Q^2$; \quad $\mathsf{sk} \leftarrow s$

$\mathsf{Enc}(\mathsf{pk},m)$: for $m \in P$ \\
$u \leftarrow \chi_s$; \quad $e_0, e_1 \leftarrow \chi_e$\\
$\mathrm{ct} \leftarrow (c_0, c_1) := (p_0 u + t e_0 + m,\; p_1 u + t e_1) \in R_Q^2$

$\mathsf{Dec}(\mathsf{sk},c =(c_0,c_1))$: \\
$\mu \leftarrow [c_0 + c_1 \cdot s]_Q$; \quad output $\mu \bmod t$
\end{algorithm}
We write $\mathsf{Enc}(\mathsf{pk},m;\xi)$ when the encryption randomness $\xi=(u,e_0,e_1)$ is made explicit, and $\Omega_{\mathsf{Enc}}=\chi_s\times\chi_e^2$ for the distribution it is drawn from.
A pair $c = (c_0,c_1) \in R_Q^2$ is a valid encryption of $m\in P$ when we have $c_0+c_1s\equiv m+te \pmod
Q$ for some small \emph{invariant noise} $e\in R$, and decryption is correct as long as $\lVert m+te\rVert_\infty<Q/2$. Adding ciphertexts adds plaintexts and adds noise; multiplying by a public $\alpha\in R$ componentwise, $\mathsf{CMul}(\alpha,c)=(\alpha c_0,\alpha c_1)$, multiplies both the plaintext and the noise term by $\alpha$. %
\paragraph{Representation of the plaintext ring.}
We set the plaintext modulus $t$ to be a prime with $t\equiv 1 \pmod{2N_0}$. Then the polynomial $X^{N_0}+1$ splits into $N_0$ distinct linear factors modulo $t$, and thus the Chinese remainder theorem gives a ring isomorphism $\mathbb{Z}_t[X]/(X^{N_0}+1)\cong
\mathbb{Z}_t^{\,N_0}$~\cite{DBLP:journals/dcc/SmartV14}, and so the plaintext ring can be reinterpreted as
\begin{equation}\label{eq:crt}
  P \;\cong\; \bigl(\mathbb{Z}_t[Y]/(Y^h)\bigr)^{N_0}.
\end{equation}
Each factor on the right is a local ring with maximal nilpotent ideal $(Y)$,
and we associate one factor with each identifier.  Identifier $\mathsf{id}_i$
is assigned slot $i$, and $\mathsf{sl}(i)\in
\mathbb{Z}_t[X]/(X^{N_0}+1)$ denotes the plaintext that is $1$ in slot $i$
and $0$ elsewhere.  We write $\mathsf{SlotDecode}$ for the CRT map of
Equation~\eqref{eq:crt}, applied blockwise: on the depth-$d$ block it is the
isomorphism $\mathbb{Z}_t[X]/(X^{N_0}+1)\cong\mathbb{Z}_t^{N_0}$. Since this is a ring isomorphism, addition acts on slots coordinatewise and
independently, so a merged tag holds in slot $i$ precisely the count of appearances of $\mathsf{id}_i$ in the merged ciphertext. The only constraint needed to keep the count accurate now is that no slot wraps modulo $t$, which is the condition $t>2^{h-1}$. %
Next we describe in sequence the various algorithms in this instantiation of our tracing scheme.
\paragraph{The algorithms.}
The setup algorithm $\mathsf{Setup}$ (Alg.~\ref{alg:rlwe-setup}) fixes the ring and generates a BGV
key pair for tag encryption, along with a signature key pair for authenticating freshly issued tags. The parameters $N_0,t,Q$ are chosen jointly: $N_0$ and the error distributions determine Ring-LWE security, while $Q$ must exceed twice the largest plaintext-plus-noise norm permitted by the deployment's fan-in and correctness horizon. Section~\ref{sec:rlwe-params} gives the constraints in full, together with concrete parameters.
\begin{algorithm}[ht!]
\caption{Setup.}
\label{alg:rlwe-setup}
\Alg{$\setup(1^\secparam,n,h)\rightarrow(\pp,\sktrace)$}{
  Choose $N_0$ a power of two with $N_0\ge n$\;
  Choose a suitably large prime $t$ with $t\equiv 1 \pmod{2N_0}$\;
  Choose $Q$ with $Q\equiv 1\pmod{2N_0}$ and $\gcd(t,Q)=1$, sized per Section~\ref{sec:rlwe-params}\;
  $R\leftarrow\mathbb{Z}[X,Y]/(X^{N_0}+1,Y^h)$\tcp*{$P=R/tR$, ciphertexts in $R_Q^2$}
  $(\mathsf{pk_{enc}},\mathsf{sk_{enc}})\leftarrow\mathsf{KeyGen}(R,Q,t)$\;
  $(\mathsf{pk_{sig}},\mathsf{sk_{sig}})\leftarrow\mathsf{SigKeyGen}(1^\lambda)$\;
  Assign slot $i$ to identifier $\mathsf{id}_i$ for each $i\in[n]$\;
  $\pp\leftarrow(R,Q,t,h,n,N_0,\mathsf{pk_{enc}},\mathsf{pk_{sig}})$,\quad
  $\sktrace\leftarrow(\mathsf{sk_{enc}},\mathsf{sk_{sig}})$\;
  \KwRet{$(\pp,\sktrace)$}\;
}
\end{algorithm}
Both tracing bounds are auditable from $\pp$. The relation $Y^h=0$ is part of the public ring and fixes the hop budget, while the degree of $X^{N_0}+1$ fixes the identifier-slot capacity $n_{\max}=N_0$; any entity receiving $\pp$ can check both. Changing $h$ changes the ring and therefore requires fresh public parameters. Unlike exponential ElGamal, however, the hop budget does not enter the hardness assumption directly: by Theorem~\ref{thm:rlwe-equiv}, security is governed by the base-ring dimension $N_0$ and the modulus $Q$. It enters indirectly, since $t>2^{h-1}$ raises the plaintext modulus and hence $Q$, so holding the security level fixed calls for a larger $N_0$ in turn; Section~\ref{sec:rlwe-params} makes this dependence explicit. A larger $h$ also increases tag size and arithmetic cost.
The tag algorithm $\mathsf{Tag}$ (Alg.~\ref{alg:rlwe-tag}) realizes the budget $h'\le h$ by issuing at depth $h-h'$, exactly as in the two previous instantiations. As before, the authority signs the ciphertext to anchor its ledger history, and $\mathsf{Verify}$ (Alg.~\ref{alg:rlwe-verify}) checks this signature.
\begin{figure}[ht!]
\centering
\begin{minipage}[t]{0.49\textwidth}
\begin{algorithm}[H]
\caption{Tag issuance.}
\label{alg:rlwe-tag}
\Alg{$\genwatermark(\sktrace,\tid_i,b,h')\rightarrow(w,\sigma)$}{
  $m\gets Y^{\,h-h'}\cdot\mathsf{sl}(i)\cdot b$\;
  $w\gets\mathsf{Enc}(\mathsf{pk_{enc}},m)$\;
  $\sigma\gets\mathsf{Sign}(\mathsf{sk_{sig}},w)$\;
  \KwRet{$(w,\sigma)$}\;
}
\end{algorithm}
\end{minipage}\hfill%
\begin{minipage}[t]{0.49\textwidth}
\begin{algorithm}[H]
\caption{Verification.}
\label{alg:rlwe-verify}
\Alg{$\mathsf{Verify}(w,\sigma) \rightarrow \{0,1\}$}{
  \KwRet{$\mathsf{Ver}(\mathsf{pk_{sig}},w,\sigma)$}\;
}
\end{algorithm}
\end{minipage}
\end{figure}
\FloatBarrier
The degrade algorithm $\mathsf{Degrade}$ (Alg.~\ref{alg:rlwe-degrade}) multiplies both ciphertext components by $Y$ and rerandomizes, implementing $m\mapsto Ym$ in $P$. It returns the randomness of the fresh encryption of zero as $\xi$. The algorithm uses no secret and, by Lemma~\ref{lem:noisefree}, adds no noise beyond that rerandomization, whose level Section~\ref{sec:rlwe-params} fixes. The merging algorithm $\mathsf{Merge}$ (Alg.~\ref{alg:rlwe-merge}) adds ciphertexts componentwise in $R_Q^2$, thereby adding their plaintexts in $P$. Its correctness is enforced by the transaction statement.
\begin{figure}[ht!]
\centering
\begin{minipage}[t]{0.49\textwidth}
\begin{algorithm}[H]
\caption{Degradation.}
\label{alg:rlwe-degrade}
\Alg{$\degrade(w)\rightarrow(w',\xi)$}{
  Parse $w=(c_0,c_1)$\;
  $\xi\xleftarrow{\$}\Omega_{\mathsf{Enc}}$\;
  $w'\gets(Y c_0,\,Y c_1)+\mathsf{Enc}(\mathsf{pk_{enc}},0;\xi)$\;
  \KwRet{$(w',\xi)$}\;
}
\end{algorithm}
\end{minipage}\hfill%
\begin{minipage}[t]{0.49\textwidth}
\begin{algorithm}[H]
\caption{Merging.}
\label{alg:rlwe-merge}
\Alg{$\mergewatermarks(w_1,\dots,w_m)\rightarrow w'$}{
  $w'\gets\sum_{j=1}^{m}w_j$\;
  \KwRet{$w'$}\;
}
\end{algorithm}
\end{minipage}
\end{figure}
\FloatBarrier
The tracing algorithm $\mathsf{Trace}$ (Alg.~\ref{alg:rlwe-trace}) decrypts and then checks the result: nonempty block indices record the degradation depth of various trackers, and nonempty slots record the tracked identifier. A tag decrypting to 0 denotes either a dummy or a fully expired tag. Tracing costs one decryption together with $h$ inverse transforms modulo $t$, hence $O(N\log N)$ operations in total, and unlike Section~\ref{sec:elgamal} its cost does not depend on how many identifiers are being tracked in total.
\begin{algorithm}[ht!]
\caption{Tracing.}
\label{alg:rlwe-trace}
\Alg{$\trace(\sktrace,w)\rightarrow S$ or $\bot$}{
  $v\leftarrow\mathsf{Dec}(\mathsf{sk_{enc}},w)\in P$\;
  \lIf{$v=0$}{\KwRet{$\bot$}}
  Write $v=\sum_{d=0}^{h-1}v_dY^d$ with $v_d\in\mathbb{Z}_t[X]/(X^{N_0}+1)$\;
  \ForEach{$d$ with $v_d\ne 0$}{
    $(c_{d,1},\dots,c_{d,n})\leftarrow\mathsf{SlotDecode}(v_d)$\;
    \lIf{$c_{d,i}>2^{h-1}$ for some $i\in[n]$}{\KwRet{$\bot$}}
    $S_d\leftarrow\{\mathsf{id}_i : i\in[n],\, c_{d,i}\ne 0\}$\;
  }
  \KwRet{$\bigcup_d S_d$}\;
}
\end{algorithm}
Finally, we establish the six properties of Section~\ref{subsec:security-defs} for this instantiation. Unforgeability relies on EUF-CMA security of the issuance signature; the other security arguments rely on Ring-LWE over $R$ (Definition~\ref{def:rlwe}). Decryption is correct only for tags whose history stays within the correctness horizon of \setup, that is, the bounds on transitions and fan-in of Section~\ref{sec:rlwe-params}; merge correctness and unforgeability are stated relative to this horizon.
\begin{theorem}[Trace correctness, Ring-LWE]
\label{thm:rlwe-trace}
For any $\mathsf{id}_i$ with $i\in[n]$, any per-tag budget $h'\le h$ and any $k<h'$, let $(w,\sigma)\leftarrow\mathsf{Tag}(\mathsf{sk_{trace}},\mathsf{id}_i,1,h')$ and let $w'$ be obtained by applying $\mathsf{Degrade}$ exactly $k$ times.
Then $\mathsf{id}_i\in\mathsf{Trace}(\mathsf{sk_{trace}},w')$ with probability $1-\mathsf{negl}(\lambda)$.
\end{theorem}
\begin{proof}
Componentwise multiplication by $Y$ multiplies the plaintext by $Y$ in $P$, and adding an encryption of zero leaves it unchanged, so after $k$ degradations the tag encrypts $Y^{\,h-h'+k}\mathsf{sl}(i)$. Since $k<h'$ we have $h-h'+k<h$, so this element is nonzero: the corresponding block at depth $h-h'+k$ is $\mathsf{sl}(i)$ and all other blocks are zero. Lemma~\ref{lem:noisefree} shows that multiplication by $Y$ does not increase noise. The bound on $Q$ from Section~\ref{sec:rlwe-params} covers the at most $h$ live rerandomization terms except with negligible probability, so decryption returns $v=Y^{\,h-h'+k}\mathsf{sl}(i)$. Slot decoding of $v_{h-h'+k}$ yields $c_{d,i}=1$ and $c_{d,i'}=0$ for $i'\ne i$, so $\mathsf{Trace}$ returns $\mathsf{id}_i$.
\end{proof}
\begin{theorem}[Merge correctness, Ring-LWE]
\label{thm:rlwe-merge}
For $i\in[m]$, let $w_i\leftarrow\mathsf{Tag}(\mathsf{sk_{trace}},\mathsf{id}_i,b_i,h_i')$ and suppose that $w_i$ is degraded $k_i\ge0$ times before merging. Write $L=\{i:b_i=1\ \wedge\ k_i<h_i'\}$ and $\tau_i=(h-h_i')+k_i$ for $i\in L$. Assume the multiplicity condition of Definition~\ref{def:merge-correctness}. If $w'\leftarrow\mathsf{Merge}(w_1,\dots,w_m)$ and the history of $w'$ stays within the correctness horizon, then $w'$ encrypts $\sum_{i\in L}Y^{\tau_i}\mathsf{sl}(\mathsf{id}_i)$, where $\mathsf{sl}(\mathsf{id}_i)$ denotes the slot assigned to the identifier of input $i$, and
$$\Pr[\mathsf{Trace}(\mathsf{sk_{trace}},w')=S_L]\ge 1-\mathsf{negl}(\lambda),$$
with $S_L$ as in Definition~\ref{def:merge-correctness}.
\end{theorem}
\begin{proof}
Every dummy input contributes zero, as does every input with $k_i\ge h_i'$
because its plaintext contains a factor $Y^h$.  Componentwise addition
therefore gives the stated plaintext $v$, and within the correctness horizon decryption returns $v$ except with negligible probability. If $L=\emptyset$, then $v=0$ and
$\mathsf{Trace}$ returns $\bot=S_L$, so assume $L\ne\emptyset$. Fix a depth $\tau$ and let
$c_{\tau,j}$ denote the number of live inputs with identifier $\mathsf{id}_j$ at that depth, so that the block $v_\tau$ equals $\sum_{j=1}^n c_{\tau,j}\,\mathsf{sl}(j)$ and slot $j$ in it holds $c_{\tau,j}$ reduced modulo $t$. We now observe that: (i) for a given block, its contents can only be driven up in depth till it degrades to zero, but not otherwise modified, and (ii) for a given slot, the value held in it can only be driven up in multiplicity by merging (till it wraps around mod $t$), but not otherwise modified.
Two separations act independently here. The blocks are module coordinates: multiplication by $Y$ shifts them and addition never carries between them, so contributions at different depths cannot collide. The slots are components of the ring isomorphism~\eqref{eq:crt}: addition acts coordinatewise, so contributions of different identifiers cannot collide either.
It therefore suffices that no slots wrap, that is, $c_{\tau,j}<t$ for every $j\in[n]$. The stated multiplicity condition gives $c_{\tau,j}\le2^{h-1}<t$. Hence $\mathsf{SlotDecode}(v_\tau)$ returns the exact multiplicities and uniquely determines $S_\tau=\{\mathsf{id}_i:i\in L,\ \tau_i=\tau\}$. Taking the union over all depths yields $S_L$.
\end{proof}
Note that no analogue of the constraint $r^n<q$ is required: only the per-slot multiplicity is bounded, so capacity and multiplicity are decoupled and the plaintext modulus can stay small however many identifiers are traced at once.
\begin{theorem}[Indistinguishability, Ring-LWE]
\label{thm:rlwe-ind}
The Ring-LWE-based tracing scheme satisfies indistinguishability assuming that $\mathrm{RLWE}_{R,Q,\chi_s,\chi_e}$ (Definition~\ref{def:rlwe}) holds.
\end{theorem}
\begin{proof}
The algorithm $\mathsf{Tag}$ returns $(w,\sigma)$ where $w=
\mathsf{Enc}(\mathsf{pk_{enc}},Y^{\,h-h'}\mathsf{sl}(i))$ when $b=1$, and returns an encryption of $0$ when $b=0$. Suppose there is a PPT distinguisher $\mathcal{D}$ that distinguishes the two with advantage $\epsilon$. Construct an IND-CPA adversary $\mathcal{B}$: on input $\mathsf{pk_{enc}}$, generate an independent signing key pair
$(\mathsf{pk_{sig}},\mathsf{sk_{sig}})$, submit the challenge messages $(Y^{\,h-h'}\mathsf{sl}(i),0)$, receive the challenge ciphertext $w$, compute $\sigma\leftarrow\mathsf{Sign}(\mathsf{sk_{sig}},w)$ and forward $(w,\sigma)$ to $\mathcal{D}$. Since $\mathsf{sk_{sig}}$ is independent of the challenge, the simulated pair is distributed exactly as in the real game and $\mathcal{B}$
preserves the advantage $\epsilon$. It remains to note that BGV is IND-CPA secure under Definition~\ref{def:rlwe}. The public key $(-(as+te),a)$ is a Ring-LWE sample up to relabelling: since $\gcd(t,Q)=1$, the bijection $(a,b)\mapsto
(t^{-1}a,\,t^{-1}b)$ of $R_Q^2$ carries $(a,as+te)$ to $(a',a's+e)$ with $a'$ uniform, and negation preserves uniformity. Given a uniform $(p_0,p_1)$, the ciphertext $(p_0u+te_0+m,\;p_1u+te_1)$ consists of two Ring-LWE samples with secret $u$, shifted by the known constant $(m,0)$, hence is indistinguishable from uniform. Combining these observations yields $\epsilon\le\mathsf{negl}(\lambda)$ since the $\mathrm{RLWE}_{R,Q,\chi_s,\chi_e}$ assumption holds.
\end{proof}
\begin{theorem}[Ephemeral traceability, Ring-LWE]
\label{thm:rlwe-eph}
For any $\mathsf{id}_i$ with $i\in[n]$ and any per-tag budget $h'\le h$, if $(w,\sigma)\leftarrow\mathsf{Tag}(\mathsf{sk_{trace}},\mathsf{id}_i,1,h')$ and $w^{(k)}$ is obtained by applying $\mathsf{Degrade}$ at least $h'$ times, then the plaintext of $w^{(k)}$ is identically $0\in P$, and
$\mathsf{Trace}(\mathsf{sk_{trace}},w^{(k)})=\bot$ with probability $1-\mathsf{negl}(\lambda)$ under the setup correctness bound.  The plaintext-erasure claim itself is unconditional and independent of $N_0$, $Q$, and $t$.
\end{theorem}
\begin{proof}
The tag is issued at depth $h-h'$, so after $k\ge h'$ degradations it encrypts $Y^{\,h-h'+k}\mathsf{sl}(i)$ with $h-h'+k\ge h$, and $Y^h=0$ in $R$. Zero is a fixed point of $m\mapsto Ym$, so all further degradations leave the plaintext at $0$. Erasure is therefore information-theoretic. Under the setup noise bound, decryption returns $0$ with overwhelming probability and $\mathsf{Trace}$ returns $\bot$.
\end{proof}
\begin{theorem}[Unforgeability, Ring-LWE]
\label{thm:rlwe-unforge}
Under EUF-CMA security of the signature scheme, the Ring-LWE instantiation satisfies unforgeability for histories within the correctness horizon: every tag with such a valid history traces only to identifiers issued by the authority. A deployment inherits unforgeability when its verification mechanism accepts only tags possessing such a history.
\end{theorem}
\begin{proof}
A reduction against EUF-CMA simulates the $\genwatermark$ oracle with the signing oracle and its own encryption keys. Thus, for every $(\pp,w^*,H)\in R_{\mathsf{hist}}$, each issued tag in $H$ was, except with negligible probability, produced by the oracle and encrypts $Y^{\,h-h'}\mathsf{sl}(i)$ with $\mathsf{id}_i\in\mathcal{Q}$, or zero. Induction along $H$ shows that $w^*$ encrypts a sum $\sum_j Y^{\tau_j}\mathsf{sl}(i_j)$ with every $\mathsf{id}_{i_j}\in\mathcal{Q}$: degradation multiplies by $Y$, a payment copies the degraded sender contribution into at most two outputs, ingestion merges one issued contribution at its issuance depth, and no tag is consumed twice. Since each identifier is issued with $b=1$ at most once, each slot has multiplicity at most $2^{h-1}<t$ at every live depth, so no slot wraps. Within the correctness horizon, decryption returns this plaintext except with negligible probability, so tracing returns only identifiers in $\mathcal{Q}$.
\end{proof}
\begin{theorem}[Unlinkability, Ring-LWE]
\label{thm:rlwe-unlink}
Under $\mathrm{RLWE}_{R,Q,\chi_s,\chi_e}$ (Definition~\ref{def:rlwe}), the Ring-LWE instantiation satisfies unlinkability. In a ledger deployment, zero knowledge of the host transaction proof additionally hides the consumed account states.
\end{theorem}
\begin{proof}
For unlinkability, the adversary submits $(w_0,w_1)$ and receives $w'\leftarrow\mathsf{Degrade}(w_b)$ for a uniform bit $b$.

\emph{Tag step.} Degradation computes $w'=Y\cdot w_b+\mathsf{Enc}(\mathsf{pk_{enc}},0;\xi)$, where $\xi$ is sampled freshly and independently of everything the adversary has seen, in particular of the internal randomness of $w_0$ and $w_1$, which it obtained from its oracle. By the argument of
Theorem~\ref{thm:rlwe-ind}, a fresh $\mathsf{Enc}(\mathsf{pk_{enc}},0;\xi)$ is computationally indistinguishable from a uniform element $U$ of $R_Q^2$ under $\mathrm{RLWE}_{R,Q,\chi_s,\chi_e}$. Replacing it accordingly gives $w'=Y\,w_b+U$, which is a uniform element of $R_Q^2$ regardless of the value of $b$, so in this hybrid the adversary's view is identically distributed for $b=0$ and $b=1$ and its guess is correct with probability exactly $\tfrac12$.

Any adversary guessing $b$ from $(w_0,w_1,w')$ with non-negligible advantage therefore yields a distinguisher against $\mathrm{RLWE}_{R,Q,\chi_s,\chi_e}$. Hence $\Pr[b'=b]\le\tfrac12+\mathsf{negl}(\lambda)$.
\end{proof}
\subsection{Hardness of Ring-LWE over the extended ring}
\label{sec:rlwe-hardness}
The ring $R$ has nilpotent elements and is therefore not the ring of integers of a number field, so the RLWE assumption over $R$ might appear to require a new assumption, much as Definition~\ref{def:ddh-prime-power} does in Section~\ref{sec:elgamal}. We show that this is not the case: Ring-LWE over $R$ is \emph{equivalent} to standard Ring-LWE over $R_0$. Since $R_0$ is cyclotomic, the assumption over $R_0$ is the most standard form used in the literature, and in particular admits reductions to worst-case problems over ideal lattices. This both supports the starting assumption, and hence the security of this instantiation, and lets us use the standard Ring-LWE literature to estimate parameters.

In a broad sense our assumption can also be viewed as a variant of module LWE, the difference being that $R$ is a particular module with specific structure over $R_0$, whereas module LWE considers a free module sampled uniformly over a base ring, and is much closer to standard LWE in that sense. The analogue of the evaluation map of Bootland, Castryck, and Vercauteren discussed above is, in our case, the projection $\pi:Y\mapsto0$, and it forms the basis of the reduction below.

In the following we assume that the distributions $\chi_s$ and $\chi_e$ are \emph{blockwise independent}, that is, independent across the $h$ blocks of $R$, with the marginal distribution of each block equal to the corresponding distribution over $R_0$. This is a natural assumption, satisfied both by the standard discrete Gaussian over $R$, which is a product of discrete Gaussians over the blocks, and by the uniform ternary distribution.

We start with an observation about the structure of $R$.
\begin{lemma}[Structure]\label{lem:structure}
Recall that $R = R_0[Y]/(Y^h)$, and that as an $R_0$-module, $R$ is of rank $h$ with basis $1,Y,\dots,Y^{h-1}$. In this basis, multiplication by any $u\in R$ is captured by the matrix
\[
M(u)\;=\;
\begin{pmatrix}
u_0 & & & \\
u_1 & u_0 & & \\
\vdots & \ddots & \ddots & \\
u_{h-1} & \cdots & u_1 & u_0
\end{pmatrix}
\in R_0^{h\times h},
\]
which is a lower triangular Toeplitz matrix over $R_0$ with every diagonal entry equal to $u_0$. In particular, for $u,v\in R$, the product $w = uv$ is represented blockwise as $w_k = \sum_{i+j=k} u_i v_j$ for $k=0,\dots,h-1$, where $u_i,v_i,w_i\in R_0$ are the blocks of $u,v,w$ respectively.
\end{lemma}
\begin{lemma}[Projection]
\label{lem:proj}
The map $\pi:R\to R_0$ sending $\sum_j u_jY^j$ to $u_0$ is a surjective ring homomorphism with kernel $(Y)$. It maps the uniform distribution over $R_Q$ to the uniform distribution over $R_0/QR_0$, and it carries an RLWE sample over $R$ with secret $s$ to one over $R_0$ with secret $s_0$ while preserving the error distribution:
\[
  \pi(a\cdot s+e)\;=\;a_0s_0+e_0 ,
\]
where $a_0$ is uniform and $e_0$ has the blockwise error distribution induced by $\chi_e$.
\end{lemma}
\begin{proof}
The map $\pi$ is reduction modulo the ideal $(Y)$. It is hence a ring homomorphism, and is surjective with kernel $(Y)$. Under the identification $R_Q\cong (R_0/QR_0)^h$ it is the projection onto the first coordinate, so it preserves uniformity. Applying it to an RLWE sample and using $\pi(as)=\pi(a)\pi(s)$ gives the displayed identity. The marginals for $a_0$, $s_0$ and $e_0$ are the per-block distributions and are independent, since the original distributions are blockwise independent.
\end{proof}
\begin{lemma}[Block peeling]
\label{lem:peel}
Given oracle access to a solver for search Ring-LWE over $R_0$ in dimension $N_0$ with modulus $Q$, the full secret $s\in R_Q$ of any Ring-LWE instance over $R$ is recoverable with $h$ oracle calls and $O(mh^2)$ multiplications in $R_0/QR_0$, where $m$ is the number of samples.
\end{lemma}
\begin{proof}
By Lemma~\ref{lem:structure}, the $k$th block of an RLWE sample in $R$ is $b_k=\sum_{i+j=k}a_is_j+e_k$. For $k=0$ this is $a_0s_0+e_0$, and one oracle call returns $s_0$ by Lemma~\ref{lem:proj}. Inductively, once $s_0,\dots,s_{k-1}$ are known, the quantity
\[
  b_k-\sum_{i=1}^{k}a_is_{k-i}\;=\;a_0s_k+e_k
\]
is computable from public data and is again a standard RLWE sample in $R_0/QR_0$, with the same uniform $a_0$ and a fresh unknown $s_k$; one further call returns $s_k$. After $h$ steps the whole secret is known, and the subtractions cost $\sum_{k<h}mk=O(mh^2)$ multiplications.
\end{proof}
We now state the main hardness equivalence.
\begin{theorem}[Hardness of Ring-LWE over $R$]
\label{thm:rlwe-equiv}
Let $\chi_s,\chi_e$ be the standard secret and error distributions for the RLWE assumption over $R_0$, and let $\Phi_s, \Phi_e$ be their $h$-blockwise repetitions, that is, blockwise independent with blockwise marginals $\chi_s$ and $\chi_e$ respectively. Then $\mathrm{RLWE}_{R, Q, \Phi_s, \Phi_e}$ over $R=R_0[Y]/(Y^h)$ with modulus $Q$ is hard if and only if $\mathrm{RLWE}_{R_0, Q, \chi_s, \chi_e}$ over $R_0$ at dimension $N_0$ with modulus $Q$ is hard. Specifically, \\{\upshape(i)} any distinguisher against $\mathrm{RLWE}_{R_0, Q, \chi_s, \chi_e}$ yields one against $\mathrm{RLWE}_{R, Q, \Phi_s, \Phi_e}$ with the same advantage; and \\ {\upshape(ii)} any distinguisher against $\mathrm{RLWE}_{R, Q, \Phi_s, \Phi_e}$ with advantage $\epsilon$ yields one against $\mathrm{RLWE}_{R_0, Q, \chi_s, \chi_e}$ with advantage $\epsilon/h$.
\end{theorem}
\begin{proof}
\emph{(i)} Given samples from the $\mathrm{RLWE}_{R, Q, \Phi_s, \Phi_e}$ challenge over $R$, apply $\pi$ componentwise to each and invoke the distinguisher for $R_0$. By Lemma~\ref{lem:proj}, RLWE samples map to samples from $\mathrm{RLWE}_{R_0, Q, \chi_s, \chi_e}$ and uniform samples map to uniform ones, so the reduction preserves advantage.

\emph{(ii)} We use a hybrid argument. For $\ell=0,1,\dots,h$ let the $\ell$th hybrid $H_\ell$ be the distribution on $m$ pairs obtained as follows:
\begin{enumerate}
  \item sample $s_0,\dots,s_{\ell-1}\leftarrow\chi_s$ honestly;
  \item then, independently for each sample, draw $a\xleftarrow{\$}R_Q$;
  \item compute the blocks of $b$ with indices $k<\ell$ as $\sum_{p+j=k}a_ps_j+e_k$ with $e_k\leftarrow\chi_e$;
  \item draw the blocks with indices $k\ge\ell$ uniformly from $R_0/QR_0$.
\end{enumerate}
By Lemma~\ref{lem:structure}, blocks with indices $k<\ell$ involve only $s_0,\dots,s_k$, so $H_\ell$ is well defined. One verifies that $H_h$ is the distribution $\mathrm{RLWE}_{R, Q, \Phi_s, \Phi_e}$ and that $H_0$ is uniform.

Assume for contradiction a distinguisher $\mathcal{D}$ that takes $m$ samples drawn either from $\mathrm{RLWE}_{R, Q, \Phi_s, \Phi_e}$ or uniformly over $R_Q$, and distinguishes them with advantage $\epsilon$. By convention $\mathcal{D}$ outputs 1 for RLWE samples and 0 for uniform ones. We build a reduction $\mathcal{R}^{\mathcal{D}}$ distinguishing samples from $\mathrm{RLWE}_{R_0, Q, \chi_s, \chi_e}$ from uniform ones. Consider adjacent hybrids $H_{\ell+1}$ and $H_\ell$. The reduction $\mathcal{R}^{\mathcal{D}}$ works as follows:
\begin{enumerate}
  \item Take $m$ challenge pairs $(\alpha,u)$ over $R_0$, drawn either from $\mathrm{RLWE}_{R_0, Q, \chi_s, \chi_e}$ with secret $s^*\leftarrow\chi_s$ or uniformly.
  \item Draw $s_0,\dots,s_{\ell-1}\leftarrow\chi_s$.
  \item For each of the $m$ samples set $a_0:=\alpha$ and $a_1,\dots,a_{h-1}\xleftarrow{\$}R_0/QR_0$.
  \item Fill the upper blocks $k>\ell$ uniformly.
  \item Set
    \[
      b_\ell\;:=\;u+\sum_{p=1}^{\ell}a_p\,s_{\ell-p}.
    \]
  \item Feed the resulting samples to $\mathcal{D}$ and output whatever it outputs.
\end{enumerate}
If $u=\alpha s^*+e$ then, writing $s_\ell:=s^*$, the right-hand side equals $a_0s_\ell+\sum_{p\ge1}a_ps_{\ell-p}+e=\sum_{p+j=\ell}a_ps_j+e$, which is precisely the $\ell$th block from $\mathrm{RLWE}_{R, Q, \Phi_s, \Phi_e}$, so the samples given to $\mathcal{D}$ are distributed exactly as $H_{\ell+1}$. If instead $u$ is uniform, then $b_\ell$ is uniform and independent across samples, and the inputs to $\mathcal{D}$ are distributed as $H_\ell$.

The reduction never uses the value $s_\ell = s^*$, only samples from $\mathrm{RLWE}_{R_0, Q, \chi_s, \chi_e}$, so it is well defined. Denote by $p_\ell$ the probability that $\mathcal{D}$ outputs 1 on inputs from $H_\ell$.  Replacing $\mathcal D$ by its complement if necessary, assume $p_h-p_0=\epsilon>0$. Choosing $\ell$ uniformly from $\{0,\dots,h-1\}$ and running $\mathcal{R}^{\mathcal{D}}$ on the corresponding samples gives distinguishing advantage
\[
  \frac1h\sum_{\ell=0}^{h-1}\bigl(p_{\ell+1}-p_\ell\bigr)
  \;=\;\frac{p_h-p_0}{h}\;=\;\frac{\epsilon}{h}
\]
for $\mathrm{RLWE}_{R_0, Q, \chi_s, \chi_e}$, which proves the assertion.
\end{proof}
\begin{remark}
  Hardness of RLWE over the ring $R$ thus reduces to hardness of RLWE over $R_0$, a much more standard formulation, and we use the latter as a proxy for parameter estimation in the next section. The factor $h$ loss in advantage between the two problems is immaterial in practice, since $h$ is small and at most a tiny polynomial in the security parameter, so the loss can be absorbed without undue adjustment to hardness.
\end{remark}
\subsection{Parameter selection for the Ring-LWE instantiation}
\label{sec:rlwe-params}
\paragraph{Constraints.}
Several conditions must hold simultaneously, as detailed in the scheme description above. We summarize them here, writing $d$ for the maximal fan-in at each transition and $L$ for the deployment's ciphertext-correctness horizon.  The tracing signal expires after at most $h$ degradations, but a ledger may continue to update and merge the resulting encryption of zero; thus $L$ may exceed $h$.
\begin{description}
\item[1. Identifier slots.] $n\le N_0$: each identifier occupies one slot of~\eqref{eq:crt}.
\item[2. Slot decomposition.] $t$ must be prime with $t\equiv 1\pmod{2N_0}$ so that the CRT-based slot decomposition applies, and $t>2^{h-1}$ to prevent slot wraparound under merging (Theorem~\ref{thm:rlwe-merge}).
\item[3. Ciphertext modulus.] $\gcd(t,Q)=1$ is required by the encryption scheme.  We additionally choose $Q\equiv 1\pmod{2N_0}$ to support efficient negacyclic transforms in the ciphertext ring.
\item[4. Rerandomization.] Unlinkability (Theorem~\ref{thm:rlwe-unlink}) requires that the mask $\mathsf{Enc}(\mathsf{pk_{enc}},0)$ added by $\mathsf{Degrade}$ be a fresh encryption of zero. Writing $B_{\mathsf{enc}}$ for the standard encryption noise, we take the rerandomization noise term to be $F=B_{\mathsf{enc}}$ at every hop.
\item[5. Correctness.] Decryption is correct while the invariant noise satisfies $\lVert e\rVert_\infty<Q/2t-\tfrac12$. Degradation is free by Lemma~\ref{lem:noisefree}, so the budget is spent only on merging and on the rerandomization step. The recurrence is $B_{j+1}=dB_j+F$ with $B_0=F=B_{\mathsf{enc}}$, so that over a lifetime of $L$ hops
\[
  B_L=B_{\mathsf{enc}}\Bigl(d^L+\frac{d^L-1}{d-1}\Bigr)
  \quad(\,=(L+1)B_{\mathsf{enc}} \text{ when } d=1\,).
\]
The correctness condition then gives $\log_2 Q\;\gtrsim\;1+\log_2 t+\log_2 B_L$.
\item[6. Security.] By Theorem~\ref{thm:rlwe-equiv} the parameters relevant to security are $(N_0,\log_2 Q)$. We take $N_0$ from the Homomorphic Encryption Security Standard~\cite{DBLP:journals/iacr/AlbrechtCCDGGHH19}, reproduced as Table~\ref{tab:hestd}, which tabulates for each dimension the largest modulus still costing $2^{128}$ classical operations with ternary secrets and $\sigma=3.2$.
\end{description}
\begin{table}[ht!]
\centering
\caption{Largest modulus admitting $128$-bit classical security at each ring dimension, for ternary secrets and $\sigma=3.2$~\cite{DBLP:journals/iacr/AlbrechtCCDGGHH19}. The listed values allow roughly $27$ modulus bits per $1024$ of dimension.}
\label{tab:hestd}
\small
\setlength{\tabcolsep}{4pt}
\begin{tabular}{lcccccc}
\hline
Ring dimension $N_0$ & 1024 & 2048 & 4096 & 8192 & 16384 & 32768\\
\hline
Max $\log_2 Q$ & 27 & 54 & 109 & 218 & 438 & 881\\
\hline
\end{tabular}
\end{table}
\paragraph{Asymptotics.}
We treat $h$, $L$, $d$ and $n$ as inputs to $\mathsf{Setup}$ bounded by $\mathrm{poly}(\lambda)$. For an asymptotic correctness statement, the tail cutoff for the discrete Gaussian error distribution is chosen to make failure negligible in $\lambda$. For the concrete estimates below we use the customary cutoff $B_e=\lceil6\sigma\rceil$. A fresh encryption then has estimated noise $B_{\mathsf{enc}}=B_e(1+2\delta_R)$, where $\delta_R$ captures coefficient expansion. Using the heuristic estimate $\delta_R=O(\sqrt{N_0h})$ for the sampled products in this scheme gives $B_{\mathsf{enc}}=O(\sigma\sqrt{N_0h})$. Condition (2) gives lower bounds on $\log_2 t$ from both $h$ and $N_0$; the following estimate additionally assumes that a suitable prime $t$ is selected with $\log_2 t=O(h+\log_2 N_0)$. Substituting these choices into condition (5) suggests
\[
  \log_2 Q\;=\;O\bigl(h+\log_2 N_0+L\log_2 d+\log_2(\sigma\sqrt{N_0h})\bigr),
\]
with $L\log_2 d$ replaced by $\log_2(L+1)$ when $d=1$. This is a correctness and noise estimate, not an asymptotic security theorem: the standard's table supplies concrete modulus limits at fixed security level rather than a proof that $N_0=\tilde\Theta(\lambda\log Q)$. For any selected $(N_0,Q)$, the ciphertext size is $|w|=2N_0h\lceil\log_2 Q\rceil$ bits.

A few elements of this design differ from the previous instantiations of Sections~\ref{sec:elgamal} and~\ref{sec:dj}. First, the hop budget $h$ does not affect the lattice dimension $N_0$, and enters security only through the mild loss in advantage of the reduction. Second, the number of identifier slots affects parameters only through the condition $n\le N_0$. Condition~(6) already requires a large $N_0$, so supporting up to $N_0$ slots requires no further parameter growth. In contrast, the constraints $r^n<q$ and $r^n<N$ in the two group instantiations cause the plaintext modulus, and therefore the group size, to grow with the identifier-slot capacity. Third, the merge fan-in $d$ substantially affects the modulus and tag size, unlike in the other instantiations. Decryption requires a noise budget that supports fan-in $d$ at every hop, which increases the required modulus size.
\paragraph{Concrete parameters.}
Table~\ref{tab:rlwe-params} illustrates the constraints at $h=16$, $t=65537$, $\sigma=3.2$, a security target of $128$ bits, and a correctness horizon of $L=32$ transitions.  Thus the tracing signal expires after 16 degradations, while ciphertext correctness is budgeted for 32 updates, including updates after expiry.  We consider three merge profiles: a linear payment chain ($d=1$), the two-state account model of Section~\ref{sec:preliminaries} ($d=2$), and a more aggressive regime with fan-in four ($d=4$). Every transition rerandomizes with fresh encryption noise, and each row satisfies conditions (1) to (6) under the estimates above.
\begin{table}[ht!]
\centering
\caption{Illustrative Ring-LWE parameters at $h=16$, $t=65537$, $\sigma=3.2$, and a correctness horizon of $L=32$ transitions, with fresh rerandomization at every transition. The dimension $N_0$ is sized against $\log_2Q$ by Table~\ref{tab:hestd}, whose limit is repeated for comparison; the total degree $N=N_0h$ affects tag size and arithmetic but not the base-ring security estimate.}
\label{tab:rlwe-params}
\small
\setlength{\tabcolsep}{4pt}
\begin{tabular}{lcccccc}
\hline
Profile & $N_0$ & $d$ & $\log_2 Q$ & max $\log_2 Q$ & $N$ & $|w|$ \\
\hline
Linear chain & 2048 & 1 & 36 & 54 & 32\,768 & 0.28\,MB \\
2 state model & 4096 & 2 & 65 & 109 & 65\,536 & 1.02\,MB \\
Merge-heavy & 8192 & 4 & 97 & 218 & 131\,072 & 3.03\,MB \\
\hline
\end{tabular}
\end{table}
Two observations emerge. First, tags are larger: at these settings a Ring-LWE tag occupies between $0.28$ and $3.0$\,MB, compared with $0.75$\,KB for exponential ElGamal and $4.1$\,KB for Damg{\aa}rd--Jurik at $h=10$. The Ring-LWE tags are therefore two to three orders of magnitude larger. In contrast, $n_{\max}=N_0$ provides thousands of identifier slots, compared with $28$ for the ElGamal parameters of Section~\ref{sec:elgamal} and $341$ for Damg{\aa}rd--Jurik. Supporting up to $N_0$ slots does not require a larger plaintext modulus because the security constraint already determines $N_0$. Tracing requires $O(N\log N)\approx 10^6$ operations, independently of the number of live identifiers, and is more efficient than in Section~\ref{sec:elgamal}.
\subsection{Efficiency and comparison}
\label{sec:comparison}
\label{app:comparison}
Table~\ref{tab:instantiations} compares the three instantiations under a common hop budget. Here $n_{\max}$ is the public identifier-slot capacity, $K$ is the number of live
identifiers searched during ElGamal tracing, and $N_0$ is the number of
Ring-LWE slots. For ElGamal and Damg{\aa}rd--Jurik, $r>2^{h-1}$ prevents
multiplicity carries. The table uses representative parameters targeting
$128$-bit security; exact sizes depend on the selected group, RSA modulus,
and Ring-LWE parameter set.
\begin{table}[ht!]
\centering
\caption{The three instantiations at $128$-bit security, using the smallest admissible base $r>2^{h-1}$ where positional encoding is required. The assumptions are Definitions~\ref{def:ddh-prime-power},~\ref{def:dcr}, and~\ref{def:rlwe}.}
\label{tab:instantiations}
\small
\setlength{\tabcolsep}{4pt}
\begin{tabular}{llll}
\hline
 & \textbf{Exp. ElGamal} & \textbf{Damg{\aa}rd--Jurik} & \textbf{Ring-LWE} \\
\hline
Plaintext ring & $\BZ_{q^h}$ & $\BZ_{N^{h}}$ & $(\BZ_t[Y]/(Y^h))^{N_0}$ \\
Nilpotent element & prime $q$ & modulus $N$ & indeterminate $Y$ \\
Degradation & $c \mapsto c^{q}\cdot\mathsf{Enc}(0)$ & $c \mapsto c^{N}\cdot\mathsf{Enc}(0)$ & $c \mapsto Yc + \mathsf{Enc}(0)$ \\
Digit/slot size (bits) & $\max(2\secparam,\, n\log_2 r)$ & $\log_2 N$ & $\log_2 t$ \\
Tag size (bits) & $2\max(3072,\, h\,\text{digit})$ & $(h+1) \log_2 N$ & $2 h N_0 \log_2 Q$ \\
Tag size at $h = 10$ & $0.75$ KB & $4.1$ KB & $\approx 0.2$--$2$ MB \\
Capacity $n_{\max}$ & $28$ & $341$ & $N_0$ \\
Trace & $O(h \cdot r^{K/2})$ & $O(nh)$ & $O(h N_0 \log N_0)$ \\
Assumption & DDH in $\BG$ & DCR & RLWE over $R$ \\
Changing $h$ & fresh parameters & same key & fresh parameters \\
Post-quantum & \xmark & \xmark & \cmark \\
\hline
\end{tabular}
\end{table}

All three constructions inherit a linear dependence on the hop budget from
the nilpotent template.  Indeed, nilpotency of index $h$ gives a strict chain
of plaintext subspaces
$P\supsetneq \nu P\supsetneq\cdots\supsetneq\nu^hP=\{0\}$, so the plaintext
must distinguish $h$ degradation depths.  The constructions differ mainly
in how much additional space they require for identifier separation and in
the cost of recovering those identifiers.

Exponential ElGamal minimizes the tag carried by each ledger state, but its
compactness depends on keeping both the identifier-slot capacity and the
live traced set modest.  Increasing capacity forces $q>r^n$ and, because
$q^h\mid p-1$, enlarges the ambient modulus $p$ by roughly $h$ times the
growth in $q$.  Tracing additionally performs a meet-in-the-middle search
whose cost is exponential in the number $K$ of live identifiers, even when
the public capacity is larger.  Damg{\aa}rd--Jurik instead decrypts directly
to an integer whose base-$N$ digits can be decoded in $O(nh)$ time.  Its tags
and online exponentiations are more expensive, but capacity is substantially
larger, tracing is insensitive to the number of identifiers actually present
in a merge, and the same modulus can support different hop budgets under the
standard DCR assumption.

This scalability comes with more expensive online arithmetic. Damg{\aa}rd--Jurik operates on $(h+1)\lceil\log_2N\rceil$-bit values, and the direct rerandomization $\rho^{N^h}$ uses an $h\lceil\log_2N\rceil$-bit exponent. Precomputing $\zeta=\zeta_0^{N^h}$ and using $\zeta^\rho$ with $\rho\xleftarrow{\$}[0,N)$ shortens the online exponent, but restricts the randomness to $\langle\zeta\rangle$ and consequently requires a subgroup variant of DCR. We retain the direct construction so that Theorems~\ref{thm:dj-indist} and~\ref{thm:dj-unlink} rely on standard DCR as stated.

The Ring-LWE construction removes positional encoding altogether: identifiers
occupy independent CRT slots, decoupling the capacity $N_0$ from the slot
modulus and making tracing polynomial and independent of $K$.  It is also the
only post-quantum option among the three.  These benefits come with much larger
tags and with parameters that must budget for noise growth under repeated
rerandomization and merging.  Consequently, exponential ElGamal is most
attractive when bandwidth and ledger-state size dominate and only a small
number of identifiers are live; Damg{\aa}rd--Jurik favors larger or heavily
merged traced sets under a conventional assumption; and Ring-LWE favors
post-quantum security and high capacity when substantially larger tags are
acceptable.  The measurements in Section~\ref{sec:evaluation} therefore evaluate the compact end of
this design space.

\section{Evaluation}\label{sec:evaluation}
We implemented the procedures introduced in Section~\ref{sec:elgamal} on a commodity laptop. For $h=10$ hops and 128-bit security, we obtained primes $p$ and $q$ of 3072 and 257 bits, respectively, and used $r=513$ for the identifier encoding, which gives 28 identifier slots. We implemented arithmetic in the order-$q^h$ subgroup of $\mathbb{Z}_p^*$ using RELIC. Table~\ref{tab:costs} reports means over 100 runs for individual operations and three runs for tracing; Section~\ref{app:addexperiments} gives the experimental setup and hardware.
A tag occupies $0.75$~KB, so the scheme adds only a few kilobytes to each transaction, an overhead comparable to the size of a Monero transaction~\cite{koe2020zerotomonero}. Issuance with its signature takes $38.1$~ms and $1.44$~KB; degradation with a stand-alone prototype proof takes $54.6$~ms and $2.13$~KB; merging two tags takes $0.02$~ms. The host ledger's own transaction proof adds a cost that depends on its proof system. Tracing all 28 identifier slots takes $1.33$~s. Tracing cost depends on the number $K$ of live identifiers and on their per-depth multiplicity $M$: the benchmark uses $M=1$, each identifier present once per depth, whereas the bound $O(h\,r^{K/2})$ of Section~\ref{sec:elgamal} covers the largest value $M=2^{h-1}$ that Lemma~\ref{lem:multiplicity} admits.

\subsection{Prototype measurements}
\label{app:addexperiments}
We evaluate the exponential-ElGamal
instantiation over the group $\BG\subseteq\BZ_p^*$ of order $q^h$.  A tag is
$w=(c_1,c_2)$; \genwatermark encrypts $b\,r^{i-1}$ and signs the fresh
ciphertext, \degrade computes
$(c_1^q g^\rho,c_2^q\mathsf{pk}_{\rm enc}^\rho)$, \mergewatermarks computes
$(\prod_j c_{1,j},\prod_j c_{2,j})$, and \trace decrypts and recovers the
nonzero base-$q$ depth digits.  The prototype also benchmarks fresh-tag
authentication and a stand-alone proof of the tag-transition relation.

\paragraph{Setup and parameters.}
The C implementation uses RELIC for big-integer arithmetic, modular
operations, and SHA-256, and SageMath~10.8 to generate the group parameters.
Experiments ran on one core of a 12th-generation Intel\textregistered\
Core\texttrademark\ i7-1260P.  We use $h=10$, $r=513$, a 257-bit prime $q\ge2^{2\secparam}$, and a 3072-bit
prime $p$, targeting 128-bit security.  These parameters support
$n_{\max}=28$ identifier
slots because $r^{28}<q<r^{29}$.  We initialize the parameters once,
perform one untimed warm-up, and then report the mean and sample standard
deviation over 100 runs.  Tracing is reported over three timed runs after
the warm-up: one trace takes more than a second, so a hundred consecutive traces heat
the core enough that the measurement tracks the thermal state of the laptop
rather than the cost of the operation.  Input generation and reset are
outside each timed interval.  The tracing row uses all $K=n_{\max}=28$ identifier
slots once, so its exact per-slot multiplicity bound is one.

\begin{table}[ht!]
\centering
\caption{Prototype costs, with primitive operations and transaction checks
reported individually.}
\label{tab:costs}
\begin{tabular}{lrrrr}
\hline
\textbf{Operation} & \textbf{Output (KB)} & \textbf{Mean (ms)} & \textbf{SD (ms)} & \textbf{Runs}\\
\hline
\genwatermark (tag and signature) & 1.44 & 38.1 & 0.9 & 100\\
$\mathsf{Verify}$ & --- & 15.0 & 1.0 & 100\\
\degrade & 0.75 & 29.0 & 2.4 & 100\\
\degrade with prototype proof & 2.13 & 54.6 & 0.5 & 100\\
Prototype-proof verification & --- & 34.3 & 0.3 & 100\\
\mergewatermarks, $k=2$ & 0.75 & 0.02 & 0.00 & 100\\
\mergewatermarks, $k=10$ & 0.75 & 0.23 & 0.01 & 100\\
\trace, $K=28$ & 0.75 & 1333.1 & 6.0 & 3\\
\hline
\end{tabular}
\end{table}

\section{Discussion and applications}
\label{app:discussion}
\label{sec:discussion}
\label{sec:applications}
ECT attaches encrypted identifier contributions to payment states. Each contribution degrades under a public hop policy, disappears after its assigned budget, and survives splits and merges until expiry. The tracing authority alone can recover live identifiers, while the ledger proof enforces valid tag transitions. This section discusses disclosure, delegation of degradation, and applications in permissioned payment systems such as central bank digital currencies, regulated stablecoins, and interbank payment networks. These applications combine ECT with institutional measures; the primitive alone neither identifies the current account holder nor authorizes an enforcement action.
\paragraph{On disclosure by the authority.}
Disclosure necessarily removes uncertainty for the identifiers that the authority names. It does not, however, help distinguish the status of an undisclosed challenge identifier, provided that the tracing key is not revealed. Tags are independently randomized encryptions, so auxiliary tagged or dummy samples are exactly what the tagging oracle of Definition~\ref{def:indist} provides; a standard hybrid argument then reduces any remaining advantage to Theorem~\ref{thm:indist} or~\ref{thm:dj-indist}. The same reasoning permits the disclosed list to contain false entries. This observation concerns the cryptographic view only: side information about an investigation may reveal tracing status independently of the tags.

Recall that our notion of indistinguishability in Definition~\ref{def:indist} hinges on a constrained model of interaction between a potentially malicious tracing authority and malicious users. In particular, the authority cannot reveal to the malicious users the full set of users being tracked. Here, we discuss a modification of our indistinguishability notion that accounts for full collusion between the authority and malicious users. Our modification simply enforces that the adversary is given a fixed list of ids $\mathcal{I} := \{\tid\}$ but not the tracing secret key $\sktrace$. Instead, we assume oracle access to the tagging functionality and the adversary can call the tagging oracle $\genwatermark(\cdot)$ on any of ids in $\mathcal{I}$. The challenge tag $(w, \sigma) \leftarrow \genwatermark(\sktrace, \tid', b, h)$ is given a fresh id $\tid' \notin \mathcal{I}$, and the adversary has to then distinguish if this
\paragraph{Delegating degradation.}
Degradation uses no secret: it applies the public degradation map and rerandomizes, and both operations are available to any party holding $\pp$. Which party advances a tag is therefore a deployment choice, subject to two conditions. That party must be able to update the structure binding tags to account states, and the ledger must enforce exactly the degradation schedule prescribed by its hop policy. Three deployments illustrate this flexibility.
\emph{Coin holders} degrade their own tags inside the transactions they submit; in this deployment, each outgoing transfer is one hop.
\emph{Custodial intermediaries}, such as payment service providers or exchanges holding accounts on behalf of their customers, can degrade instead. They already assemble the transactions of those accounts, so the construction is unchanged, and clients with limited resources are relieved of producing the proof. Custody grants no tracing information: the intermediary manipulates ciphertexts under the authority's key and learns neither whether a tag is real nor how much of its budget remains.
\emph{Validators} can instead degrade every published tag at a fixed block height or once per epoch. If tags and their updates are public, all parties can recompute the degradation map, so no zero-knowledge proof is needed. This changes the budget from transaction hops to ledger epochs and ages dormant accounts as well. Waiting then consumes the budget, whereas holder-driven degradation preserves a dormant tag but requires transactions to exhaust it. Ledger-driven degradation requires a publicly updatable tag structure rather than owner-controlled account states; it is therefore a different ledger model, not merely a parameter choice.
\paragraph{Beyond coin tracing.}
ECT records that a ledger state descends from a tagged identifier; turning that signal into action requires the ledger or operator to associate the state with an account and apply an external policy. Possible responses include freezing the current account, imposing transaction limits, restricting counterparties, or requesting additional records under the applicable legal process. ECT can also support the investigation of structuring by preserving the tagged contribution across splits and merges, but it does not itself classify a transaction pattern as suspicious. These measures are complementary and may be applied only while the tag is traceable, or later on the basis of evidence recorded during that window.
\paragraph{Prospective versus retrospective tracing.}
ECT is \emph{prospective} by design: $\trace$ can recover an identifier only if the authority previously encoded it with $b=1$. An untagged contribution may later share a state with tagged funds, but its own origin cannot be recovered retroactively, even with $\sktrace$. This separation limits retrospective mass surveillance. Deployments requiring retrospective tracing need an additional mechanism, for example a distinct account type whose transactions carry an opening trapdoor. Distributing that trapdoor among a threshold committee can prevent unilateral tracing, as in PEReDi~\cite{DBLP:conf/ccs/KiayiasKS22} and UTT~\cite{DBLP:journals/iacr/TomescuBAAGPY22}; such a mechanism is complementary to, rather than an extension of, ECT.
\paragraph{Issuance, renewal, and hop policy.}
ECT makes the reach of every issued contribution explicit, and a deployment can layer its accountability guarantees on top of this bound. Policy enters at issuance: the authority's signature authenticates each board entry, and a deployment that wants to enforce a warrant's shorter budget $h'<h$, or a quota on the number of tagged accounts, against the authority itself can add a zero-knowledge proof of well-formed issuance; the public slot capacity bounds how many identifiers a tag distinguishes, not how many accounts receive tags. Renewal issues a fresh identifier (Section~\ref{sec:system-model}), so every identifier has at most one live issuance and the no-carry guarantee of Lemma~\ref{lem:multiplicity} holds across repeated warrants. Under the default payment policy a hop is an outgoing transfer, so dormant funds keep their contribution until they move; deployments that also want to bound elapsed time can adopt the epoch rule, under which every ingestion advances existing contributions.
\paragraph{User-requested protective tracing.}
A user facing kidnapping, extortion, or account-compromise risk may voluntarily request traceable tags for their holdings. For example, a merchant could request $b=1$ with budget $h'=5$ and either announce the protection as a deterrent or keep it hidden using indistinguishability. If the funds are transferred under duress, the authority can inspect subsequent ledger tags and use $\trace$ to recognize descendants of the merchant's tagged state. Where the ledger links those states to enforceable accounts, an authorized party may then freeze or investigate them. After $h'$ transaction hops, ECT no longer supplies that tracing signal.
\paragraph{Time-bounded regulatory investigation.}
A court order may authorize prospective tracing of a suspicious account for at most $h'$ transaction hops. The authority issues tags with $h'\le h$, after which degradation removes the encoded identifier without requiring the authority to revoke it. This enforces a transaction-distance bound, not a wall-clock deadline; a time-based warrant instead requires the epoch-driven variant above or an additional ledger rule.
\paragraph{Interbank tracing.}
In an interbank payment network, a financial intelligence unit or similar authority may tag funds after a participating bank reports a suspicious account. Each bank carries the encrypted tag in its authenticated payment state and enforces the agreed transition policy when funds move within or across banks. The authority can then recognize tagged funds for the prescribed number of hops, while the receiving bank can use its customer records to identify the current account holder and take any authorized action. Banks learn no tracing information from the tag itself, and unrelated customers' transaction histories remain hidden. This deployment requires an interoperable payment rail that binds tags to states, but does not require banks to pool their transaction records.
\paragraph{Privacy budget integration.}
ECT can be combined with a ledger-enforced privacy budget, such as a zero-knowledge proof that a user's monthly volume remains below a threshold. Crossing the threshold would trigger a separate policy that obtains an authority-issued traceable tag for the resulting state; implementing this trigger without revealing compliant users requires additional protocol machinery. The budget determines \emph{when} monitoring begins, while ECT bounds how far the resulting signal propagates. Because tags survive splits and merges, the combination can support investigations of structuring, although detection of the pattern remains outside the ECT primitive.
\paragraph{Temporary institutional transparency.}
Exchanges and custodians may ask the authority to issue traceable tags for designated deposits or withdrawals. The authority's signature authenticates each initial tag, while indistinguishability hides whether it encodes an identifier. Auditors can then follow the propagation of tagged funds for a bounded number of hops, provided they receive the relevant tracing result and ledger view. This offers bounded transaction-flow accountability; it does not by itself prove solvency or correct custody, which require separate accounting commitments and proofs.

\section{Concluding remarks and outlook}
\label{sec:conclusion}
We answer our opening question affirmatively: a simple and efficient protocol can support targeted tracing while cryptographically bounding its reach, even against the tracing authority itself. ECT realizes this goal through encrypted identifier contributions that follow funds through splits and merges and advance toward expiry under public degradation. After its assigned number of policy-defined hops, a contribution becomes zero and cannot be recovered even by the tracing authority, without affecting other live contributions. The authority can therefore initiate a trace but cannot extend it indefinitely.

Our nilpotent-ring template makes expiry exact and expired contributions neutral under merging. Its three instantiations offer different tradeoffs: exponential ElGamal gives the smallest tags and a deliberately small identifier capacity; Damg{\aa}rd--Jurik provides greater capacity and polynomial-time tracing; and Ring-LWE offers post-quantum security and thousands of slots at the cost of larger tags. At 128-bit security with a ten-hop budget, our ElGamal implementation adds a few kilobytes per transaction and traces 28 identifiers in 1.33~s. In every construction, public parameters expose the hop and capacity bounds.

ECT can support warrant-based tracing in CBDCs and regulated stablecoins, cross-institution tracing without pooling banks' records, and applications such as protective tracing or temporary institutional transparency. The surrounding system must still authorize issuance, define what constitutes a hop, enforce tag transitions, and determine how to act on tracing results. Since ECT cannot recover transactions that preceded tag issuance, it complements retrospective investigative tools. Sections~\ref{app:comparison} and~\ref{app:discussion} detail the construction tradeoffs and deployment choices.

Several directions remain open. A zero-knowledge proof of policy-compliant issuance would let deployments enforce warrant budgets and tagging quotas against the authority itself without revealing who is tagged. Implementations of the Damg{\aa}rd--Jurik and Ring-LWE constructions, integrated with a complete ledger proof system, would quantify their end-to-end costs. For Ring-LWE, a sharper analysis of what the residual noise of an expired tag reveals would carry the exact expiry of the group-based constructions over to the post-quantum setting.

\label{endofbody}%
\bibliographystyle{splncs04}
\bibliography{refs}
\appendix
\section{Standard cryptographic notions}\label{app:cryptodefs}
\subsection{Public key encryption}\label{app:pke}
Here, we recall public key encryption and IND-CPA security.
\begin{definition}[Public key encryption scheme]
    A public key encryption scheme over a message space $\mathbb{M}$ with security parameter $1^\lambda$ is a tuple of 3 PPT algorithms $(\mathsf{KeyGen}, \mathsf{Enc}, \mathsf{Dec})$ such that
    \begin{itemize}
        \item $\mathsf{KeyGen}(1^\lambda) \rightarrow (\mathsf{pk},\mathsf{sk})$ generates a public/secret key pair.
        \item $\mathsf{Enc}(\mathsf{pk}, m) \rightarrow c$ encrypts message $m \in \mathbb{M}$ to ciphertext $c$.
        \item $\mathsf{Dec}(\mathsf{sk}, c) \rightarrow m$ decrypts ciphertext $c$ to message $m$ or $\bot$.
    \end{itemize}
We require \emph{correctness}, i.e., for all $\lambda \in \mathbb{N}$, $(\mathsf{pk},\mathsf{sk}) \leftarrow \mathsf{KeyGen}(1^\lambda)$ and $m \in \mathbb{M}$, $\Pr[\mathsf{Dec}(\mathsf{sk}, \mathsf{Enc}(\mathsf{pk}, m)) = m] = 1$.
\end{definition}
We further require that our public-key encryption scheme is IND-CPA secure against computationally bounded adversaries.
IND-CPA security is defined by the following indistinguishability game between an adversary $\mathcal{A}$ and challenger $\mathcal{C}$.
\paragraph{IND-CPA indistinguishability game $\mathsf{PubK^{cpa}_{\mathcal{A},\Pi}}(\lambda)$.}
\begin{enumerate}
    \item $\mathcal{C}$ samples $(\mathsf{pk},\mathsf{sk}) \leftarrow \mathsf{KeyGen}(1^\lambda)$.
    \item $\mathcal{A}$ gets oracle access to $\mathsf{Enc}(pk, \cdot)$. It outputs a pair of messages $(m_0,m_1)$ of the same length and sends them to $\mathcal{C}$.
    \item $\mathcal{C}$ samples a uniform random bit $b \in \{0,1\}$, and computes the challenge ciphertext $c^* = \mathsf{Enc}(\mathsf{pk},m_b)$. $\mathcal{C}$ sends the challenge ciphertext $c^*$ to $\mathcal{A}$.
    \item $\mathcal{A}$ outputs a guess bit $b'$. The output of the game is $1$ if $b'=b$ and $0$ otherwise.
\end{enumerate}
\begin{definition}[IND-CPA security]
    A public key encryption scheme $\Pi = (\mathsf{KeyGen}, \mathsf{Enc}, \mathsf{Dec})$ is IND-CPA secure if for every PPT adversary $\mathcal{A}$, there exists a negligible function $\mathsf{negl}$ such that $$\Pr[\mathsf{PubK^{cpa}_{\mathcal{A},\Pi}}(\lambda)=1] \leq \frac{1}{2} + \mathsf{negl}(\lambda).$$
\end{definition}
\subsection{Signature scheme}\label{app:signature}
Here, we recall signature schemes and EUF-CMA security.
\begin{definition}[Signature scheme]
    A (digital) signature scheme over a message space $\mathbb{M}$ with security parameter $1^\lambda$ is a tuple of 3 PPT algorithms $(\mathsf{KeyGen}, \mathsf{Sign}, \mathsf{Ver})$ such that
    \begin{itemize}
        \item $\mathsf{KeyGen}(1^\lambda) \rightarrow (\mathsf{pk},\mathsf{sk})$ generates a public/secret key pair.
        \item $\mathsf{Sign}(\mathsf{sk}, m) \rightarrow \sigma$ signs a message $m \in \mathbb{M}$ to output a signature $\sigma$.
        \item $\mathsf{Ver}(\mathsf{pk}, m, \sigma) \rightarrow b$ takes as input a public key $\mathsf{pk}$, a message $m$, and a signature $\sigma$ and outputs a bit $b$ representing if the $\sigma$ is a valid signature on $m$ or not.
    \end{itemize}
We require \emph{correctness}, i.e., for all $\lambda \in \mathbb{N}$, $(\mathsf{pk},\mathsf{sk}) \leftarrow \mathsf{KeyGen}(1^\lambda)$ and $m \in \mathbb{M}$, $\Pr[\mathsf{Ver}(\mathsf{pk}, m, \mathsf{Sign}(\mathsf{sk}, m)) = 1] = 1$.
\end{definition}
We further require that our signature scheme is EUF-CMA secure against computationally bounded adversaries.
EUF-CMA security is defined by the following indistinguishability game between an adversary $\mathcal{A}$ and challenger $\mathcal{C}$.
\paragraph{EUF-CMA indistinguishability game $\mathsf{SigForge^{EUF-CMA}_{\mathcal{A},\Sigma}}(\lambda)$.}
\begin{enumerate}
    \item $\mathcal{C}$ samples $(\mathsf{pk},\mathsf{sk}) \leftarrow \mathsf{KeyGen}(1^\lambda)$.
    \item $\mathcal{A}$ receives $\mathsf{pk}$ and gets oracle access to $\mathsf{Sign}(\mathsf{sk},\cdot)$. $\mathcal{A}$ can repeatedly query the signing oracle $\mathsf{Sign}(\mathsf{sk},\cdot)$ with messages $(m_1, \dots, m_Q)$ of its choosing to get signatures $(\sigma_1, \dots, \sigma_Q)$.
    \item $\mathcal{A}$ now outputs a new message and signature pair $(m^*,\sigma^*)$ such that $m^* \notin (m_1, \dots, m_Q)$, and sends $(m^*,\sigma^*)$ to $\mathcal{C}$.
    \item The output of the game is $1$ if $\mathsf{Ver}(\mathsf{pk}, m^*, \sigma^*)=1$ and $m^* \notin (m_1, \dots, m_Q)$, and $0$ otherwise.
\end{enumerate}
\begin{definition}[EUF-CMA security]
    A signature scheme $\Sigma = (\mathsf{KeyGen}, \mathsf{Sign}, \mathsf{Ver})$ is EUF-CMA secure if for every PPT adversary $\mathcal{A}$, there exists a negligible function $\mathsf{negl}$ such that $$\Pr[\mathsf{SigForge^{EUF-CMA}_{\mathcal{A},\Sigma}}(\lambda)=1] \leq \mathsf{negl}(\lambda).$$
\end{definition}
\subsection{Non-interactive zero-knowledge (NIZK) proof system}\label{app:nizk}
We recall NIZK proof systems and their security.
\begin{definition}[NIZK proof system]
A pair of PPT algorithms $(\mathcal{P,\mathcal{V}})$ with security parameter $\lambda$ is a NIZK proof system for an NP relation $R$ if there exists some polynomial $\mathsf{poly}$ such that
\begin{enumerate}
    \item (Completeness.) For $x \in R \cap \{0,1\}^\lambda$, and all witnesses $w$ for $x$,
    $$\Pr[r \leftarrow\{0,1\}^{\mathsf{poly}(\lambda)}; \pi \leftarrow \mathcal{P}(r,x,w); \mathcal{V}(r,x,\pi) =1] \geq 1 - \negl(\lambda).$$
    \item (Soundness.) For all $x \in  \{0,1\}^\lambda \setminus R$, and for all algorithms $\mathcal{P}^*$:
    $$\Pr[r \leftarrow\{0,1\}^{\mathsf{poly}(\lambda)}; \pi \leftarrow \mathcal{P}^*(r,x); \mathcal{V}(r,x,\pi) =1] \leq \negl(\lambda).$$
    \item (Zero-knowledge.) There exists a PPT algorithm $\mathsf{Sim}$ such that for any $x \in R \cap \{0,1\}^\lambda$ and any witness $w$ for $x$, the following ensembles of random variables are computationally indistinguishable:
    $$(1) \  \{r \leftarrow\{0,1\}^{\mathsf{poly}(\lambda)}; \pi \leftarrow \mathcal{P}(r,x,w) \}_\lambda$$
    $$(2) \ \{(r,\pi) \leftarrow \mathsf{Sim}(x) \}_\lambda$$
\end{enumerate}
\end{definition}
\end{document}